\documentclass[10pt,conference,letterpaper]{IEEEtran}

\IEEEoverridecommandlockouts
\usepackage{float}
\restylefloat{table}
\usepackage{etoolbox}
\usepackage{titlesec}
\usepackage{adjustbox}
\usepackage{booktabs}
\usepackage{lipsum}
\usepackage{mathtools, cuted}
\usepackage{autobreak}
\usepackage[utf8]{inputenc}
\usepackage{nccmath}
\usepackage{flushend}
\usepackage{comment}
\usepackage{soul}
\usepackage{physics}
\usepackage{afterpage}
\usepackage{url}
\usepackage[hidelinks]{hyperref}
\usepackage{amsfonts,amsmath,amstext,amsthm,xspace,enumerate,graphicx}
\usepackage{enumitem}
\usepackage{algpseudocode}
\usepackage[font=footnotesize]{caption}
\DeclareCaptionFont{custom}{\fontsize{5.5}{6}\selectfont}
\usepackage[font=custom]{subcaption}
 \usepackage{hyperref}
\usepackage[commandnameprefix=always]{changes}
\usepackage[english]{babel}
\usepackage{siunitx}
\usepackage[ruled]{algorithm}
\usepackage{cases}
\usepackage{stfloats}
\usepackage{placeins} 
\usepackage{booktabs}
\usepackage{makecell}

\usepackage{subcaption} 
\usepackage{caption}

\theoremstyle{plain}
\newtheorem{theorem}{Theorem}[section]
\newtheorem{definition}{Definition}
\newtheorem*{remark*}{Remark}

\addto\extrasenglish{

}

\newcounter {optimization}
\makeatletter

\titlespacing{\subsubsection}{0pt}{1pt}{0pt}
\titlespacing{\subsection}{0pt}{1pt}{0pt}
\titlespacing{\section}{0pt}{2pt}{0pt}

\makeatletter
\def\thm@space@setup{%
  \thm@preskip=0pt
  \thm@postskip=0pt
}

\newcommand{\Hamta}[1]{{\color{red}\textbf{[(Hamta) #1]}}}

\usepackage[minnames=1,maxnames=1,style=ieee,doi=false,isbn=false]{biblatex}

\usepackage[acronym]{glossaries}
\makeglossaries
\newacronym{cp}{CP}{Content Provider}
\newacronym{ed}{ED}{Edge Caching Device}
\newacronym{mu}{MU}{Mobile User}
\newacronym{SCQ}{SCQ}{Secure Caching Quality}
\newacronym{idas}{IBAS}{ID-based aggregate signature}
\newacronym{msn}{MSN}{Mobile Social Network}
\newacronym{QoE}{QoE}{Quality of Experience}
\newacronym{rl}{RL}{Reinforcement Learning}

\usepackage{amsmath,amssymb,amsfonts}
\usepackage{graphicx}
\usepackage{textcomp}
\usepackage{xcolor}
\def\BibTeX{{\rm B\kern-.05em{\sc i\kern-.025em b}\kern-.08em
    T\kern-.1667em\lower.7ex\hbox{E}\kern-.125emX}}
\begin{document}

\title{Decentralized Edge Caching under Budget and Storage Constraints: A Game-Theoretic Approach
}

\author{\IEEEauthorblockN{Hamta Sedghani}
\IEEEauthorblockA{
\textit{Politecnico di Milano}\\
Milan, Italy \\
hamta.sedghani@polimi.it}
\and
\IEEEauthorblockN{Zahra Seyedi}
\IEEEauthorblockA{
\textit{Politecnico di Milano}\\
Milan, Italy \\
zahrasadat.seyedi@mail.polimi.it}
\and
\IEEEauthorblockN{Mauro Passacantando}
\IEEEauthorblockA{
\textit{University of Milano-Bicocca}\\
Milan, Italy \\
mauro.passacantando@unimib.it}
\and
\IEEEauthorblockN{
Danilo Ardagna}
\IEEEauthorblockA{
\textit{Politecnico di Milano}\\
Milan, Italy \\
danilo.ardagna@polimi.it}
}

\maketitle

\begin{abstract}

The rapid growth of mobile social networks (MSNs) has significantly increased the demand for low-latency and reliable content delivery, motivating the deployment of edge caching systems. In practice, multiple content providers (CPs) compete for the limited storage resources of edge devices (EDs), while facing heterogeneous budgets and operational costs. This paper investigates a decentralized multi-CP edge caching framework that jointly accounts for CP budget constraints, ED storage limitations, and strategic interactions among all entities.
We formulate the interaction between CPs and EDs as a hierarchical game, combining a Stackelberg model for CP–ED interactions with a non-cooperative game among competing CPs. Under light storage constraints, we show that CP competition constitutes an exact potential game, ensuring the existence of a pure-strategy Nash equilibrium and enabling decentralized convergence. 
When storage constraints are binding, the resulting game loses this structure; nevertheless, extensive simulations demonstrate stable and efficient convergence in practice.
Through a comprehensive numerical evaluation, we show that convergence behavior is primarily driven by CP competition rather than the scale of edge infrastructure. We further reveal that storage scarcity fundamentally alters economic outcomes, amplifying inequality among CPs while increasing the relative bargaining power of EDs. The proposed framework provides a scalable and economically grounded solution for decentralized resource allocation in multi-provider edge caching systems.
\end{abstract}


\section{Introduction}\label{sec:intro}

The rapid growth of \glspl{mu} and the explosive popularity of \glspl{msn}, such as YouTube, TikTok, and Instagram, have driven an unprecedented surge in wireless data traffic~\cite{luo2025cost}. Global mobile data traffic is expected to increase by 17\%, reaching nearly 430 exabytes per month by 2030~\cite{ericsson}, intensifying the demand for low-latency content delivery. However, repeatedly retrieving popular content from distant \glspl{cp} or centralized servers increases delay and network congestion. Edge caching addresses this issue by storing popular \gls{msn} content closer to \glspl{mu}, thereby reducing latency, alleviating congestion, improving \gls{QoE}, and lowering delivery costs~\cite{zhang2025survey, luo2025cost}.

Despite these advantages, edge caching introduces several challenges. \glspl{cp} rely on \glspl{ed} for content delivery but must balance caching benefits against constraints such as limited budget, bandwidth, and storage~\cite{chen2024dynamic}. Meanwhile, self-interested \glspl{ed} may behave strategically or maliciously, and their open nature makes them vulnerable to security threats such as DDoS and man-in-the-middle attacks~\cite{Xu2020, ma2024survey}. In addition, increasing competition among \glspl{mu} for limited resources can degrade \gls{QoE}~\cite{he2024design}. These challenges call for robust and efficient edge caching mechanisms.

Game theory has been widely used to model interactions in edge caching systems, with approaches including auctions, pricing, and collaborative strategies~\cite{yu2025auction}. However, most existing works focus on a single \gls{cp} interacting with multiple \glspl{ed}, overlooking competition among multiple \glspl{cp} \cite{khan2024content, liao2025context}.

In practice, \glspl{ed} serve multiple \glspl{cp}, leading to competition over limited storage resources. Each \gls{cp}, constrained by its budget, must strategically allocate resources to maximize its utility. This setting introduces complex interactions not only between \glspl{cp} and \glspl{ed}, but also among competing \glspl{cp}, motivating the need for game-theoretic models that capture budget constraints and shared resource competition \cite{sedghani2021}.

Motivated by Xu et al.~\cite{Xu2020}, we significantly extended their work and developed a practical and efficient secure edge caching framework for multi-\gls{cp} systems, in which multiple \glspl{cp} compete for caching services on nearby \glspl{ed} under budget constraints. We introduce a lightweight, resource-aware optimization model that bounds content- and device-level payments to enable scalable deployment in edge environments. 
Under typical operating regimes, where storage provisioning avoids frequent contention, we model \glspl{cp} competition as a potential game and the \gls{cp}–\gls{ed} interaction as a Stackelberg game, which together ensure the existence of a Nash equilibrium and enable decentralized convergence via a finite improvement dynamics. 
When storage constraints are strictly binding, the 
equilibrium existence is not guaranteed; nevertheless, extensive experimental results demonstrate that the proposed decentralized protocol converges reliably in practice and achieves strong performance in terms of cost efficiency, scalability, interaction efficiency, and robustness under dynamic network conditions.

The remainder of this paper is structured as follows: 
\autoref{sec:related} reviews related work.
In \autoref{sec:system_model}, we introduce our system model, detailing the network, content, and threat model. We then present the problem formulation in \autoref{sec:problem_formulation}. Our analysis of the optimal strategy obtained through the game solution is detailed in \autoref{sec:Stackelberg_game}. \autoref{sec:experimental} provides an evaluation of our proposed scheme. Conclusions are finally drawn in \autoref{sec:conclusions}.

\vspace{-1mm}
\section{Related Work}\label{sec:related}
Edge caching is a key technique for reducing latency and backhaul congestion by storing popular content closer to users~\cite{zhang2025survey}. Existing approaches include coded and non-coded caching, as well as proactive, cooperative, and adaptive strategies designed to improve cache efficiency under limited storage and bandwidth resources~\cite{feng2025federated, niknia2025attention}. Most studies model a single \gls{cp} interacting with multiple \glspl{ed} that provide caching services~\cite{khan2024content, liao2025context}. 
With the rise of \glspl{msn}, recent works have further incorporated user-centric features, such as mobility patterns, social relationships, and contextual information, to enhance caching decisions and content relevance~\cite{Xu2020, chaudhary2025pencache}. 
However, these models generally assume a single \gls{cp} and fail to capture realistic scenarios involving multiple competing providers.

From an optimization perspective, prior work has focused on resource pricing, allocation, delay minimization, and QoE enhancement~\cite{guo2025optimal}. Representative approaches include pricing-based caching schemes, delay-aware clustering strategies, and cost-efficient caching models that account for content freshness and budget limitations~\cite{yan2021pricing, doostmohammadi2023dynamic, abolhassani2024optimal}. While some studies incorporate \gls{cp} budget constraints into the optimization process~\cite{wang2025investment, yuan2024efficient}, they typically assume a single \gls{cp}, overlooking competition among multiple \glspl{cp} for limited edge storage resources.

To solve edge caching problems, a wide range of methods has been proposed, including mathematical optimization, auction mechanisms, game theory, and learning-based approaches~\cite{ismail2025survey}. 
Game-theoretic models, such as Stackelberg, contract-based, and non-cooperative games, are widely used to capture interactions among \glspl{cp}, \glspl{ed}, and \glspl{mu}, enabling efficient pricing and resource allocation~\cite{cheng2024stackelberg, fan2024contract, jiang2022game}. 
However, these models generally neglect direct competition among multiple \glspl{cp}. In parallel, reinforcement learning approaches have been applied to handle dynamic demand and uncertainty, improving cache hit rates and adaptability~\cite{zhong2020deep, wei2024cooperative, Xu2020}. Despite their effectiveness, RL-based methods often incur high computational and communication overhead and suffer from limited real-time efficiency, particularly in large-scale multi-\gls{cp} environments.

Although existing works have advanced edge caching through diverse strategies and methodologies, they often overlook the joint impact of multiple competing \glspl{cp}, heterogeneous budgets, and limited edge storage. To address these limitations, we propose a resource-aware edge caching framework that explicitly models multi-\gls{cp} competition under budget and capacity constraints. Our approach employs a two-layer game-theoretic formulation and a lightweight decentralized mechanism to achieve efficient and scalable resource allocation without centralized coordination.

\vspace{-1mm}
\section{System Model}\label{sec:system_model}
This section presents the system model. 
We first introduce the network and content architecture, followed by the threat model mostly by extending the work  in \cite{Xu2020}. 
For convenience, all notations used in the paper are summarized in 
Table~\ref{tab:notations}  in  \autoref{sec:Appendix}.

\subsection{Network Architecture}
\label{ss:NetworkArch}

The network architecture of the secure edge caching model is depicted in \autoref{fig:system_model}. It comprises three key entities: multiple content providers (\glspl{cp}), a set of edge caching devices (\glspl{ed}), and a group of mobile users (\glspl{mu}). The roles and responsibilities of these entities within the system are specified as follows:
\begin{figure}
    \centering
    \includegraphics[width=0.7\linewidth,height=3cm]{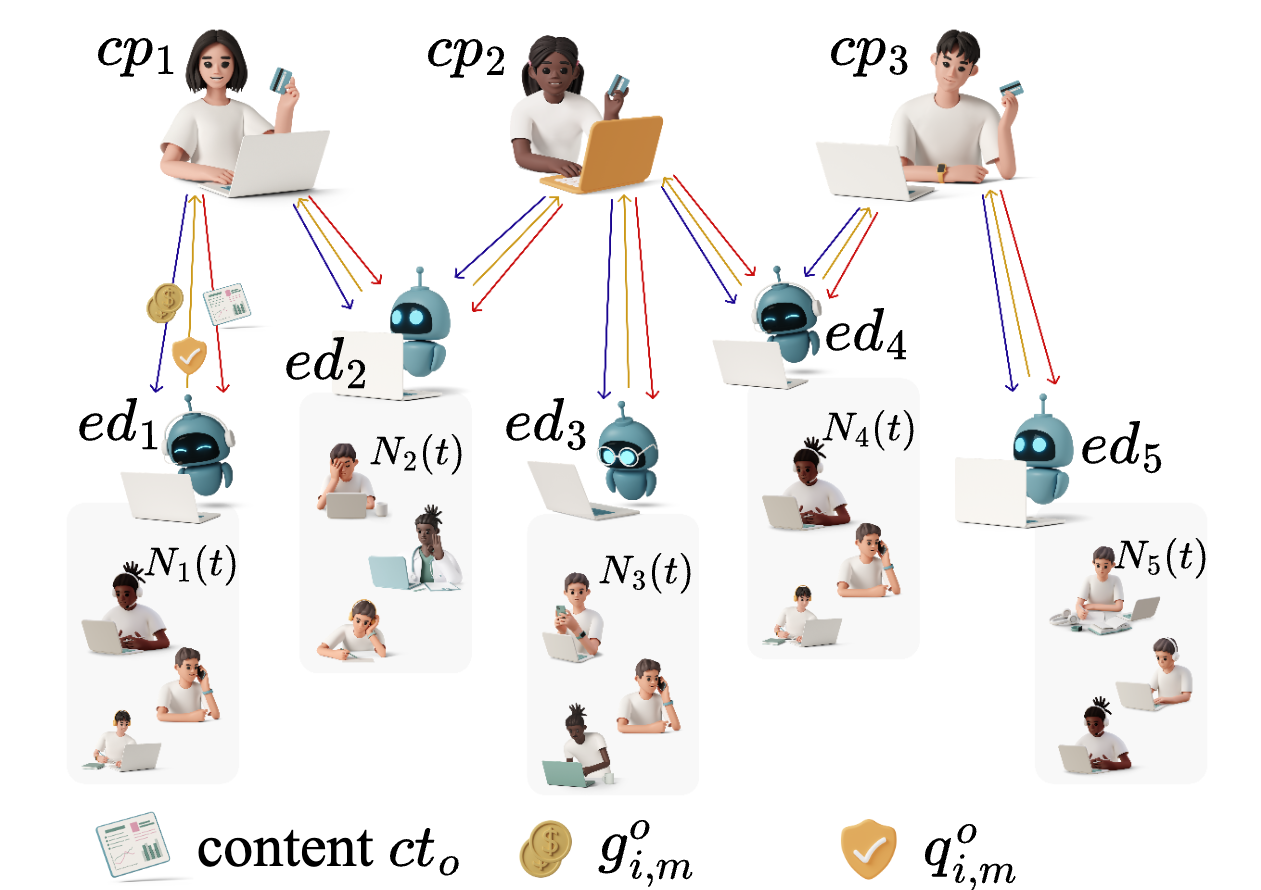}
    \caption{Network Architecture}
    \label{fig:system_model}
\end{figure}
\begin{enumerate}[leftmargin=3mm]
	\item \textbf{\glspl{cp}}: 
    Denoted as $\mathcal{O} = \{cp_1, \dots, cp_o, \dots, cp_O\}$, the \glspl{cp} serve as the origin of content (e.g., videos or files) requested by \glspl{mu}. They are typically deployed at geographically distant locations. They aim to expand their user base and maximize revenue by enabling \glspl{mu} to access popular content efficiently and securely. Nevertheless, their remote placement introduces considerable latency in content delivery, thereby degrading users’ \gls{QoE}. To mitigate this challenge, \glspl{cp} seek to securely cache popular content on edge devices, providing faster access while containing operational costs.

	\item \textbf{\glspl{ed}}: 
    The set of \glspl{ed} is represented as $\mathcal{I} = \{ ed_1, \dots, ed_i, \dots, ed_I\}$. Positioned closer to \glspl{mu} (e.g., within local environments such as schools or hospitals), these devices cache popular content to reduce access latency. Strategically deployed at the edge of the backhaul network, they enable \glspl{mu} to retrieve content from the nearest caching node with minimal delay. To enhance coverage and reduce redundancy, each \gls{ed} is installed at a distinct location. Nonetheless, \glspl{ed} exhibit certain behavioral traits, such as selfishness and susceptibility to open-access threats, which impact system performance and raise security concerns. 

    To address these issues, \glspl{ed} may employ security mechanisms, such as disaster recovery modes, that distribute content replicas across multiple locations to preserve the integrity and confidentiality of cached data. 
    By adjusting the level of secure caching, each edge device can deliver different degrees of protection for stored content. In this framework, the \gls{SCQ} offered by $ed_i \in \mathcal{I}$ for content $ct_m^o$, belonging to $cp_o$, is denoted as $q^{o}_{i,m}$, and defined as:	
	\begin{equation}
    q^{o}_{i,m}
    \begin{cases}
    = -1, & \text{if $ed_i$ is selfish,}\\
    \in [0,1], &\text{otherwise},
    \end{cases}
    \end{equation}
    where, as in \cite{Xu2020}:
	\begin{itemize}
        \item $q^{o}_{i,m} = -1$ denotes that the $ed_i$ behaves selfishly and provides manipulated or corrupted content to \glspl{mu}.
        \item $q^{o}_{i,m} = 1$ represents the maximum \gls{SCQ} provided by $ed_i$.
        \item $0 < q^{o}_{i,m} < 1$: Corresponds to an intermediate level of \gls{SCQ} provided by $ed_i$.
        \item $q^{o}_{i,m} = 0$ indicates that $ed_i$ does not engage in caching $ct_m^o$ and merely forwards the content.
    \end{itemize}

	\item  \textbf{\glspl{mu}}:
    These entities request content from nearby \glspl{ed} to reduce latency. Their mobility and interaction behavior are characterized as follows:  

    \begin{itemize}
        \item \textit{Content Requests}: \glspl{mu} obtain cached content directly from nearby \glspl{ed}, ensuring low-latency delivery. If the requested content is not cached, the \gls{ed} relays the request to the \gls{cp} or another \gls{ed}.

        \item \textit{Mobility Model}: The movement of \glspl{mu} is modeled as a random walk process. Each \gls{mu}’s velocity is uniformly distributed within $[V_{\text{min}}, V_{\text{max}}]$, and the movement direction is uniformly random over $[0, 2\pi]$. An \gls{mu} may remain stationary for a random duration within $[0, T_{\text{max}}]$ before resuming movement.


        \item \textit{Coverage and Variability}: At time slot $t$, the set of \glspl{mu} within the coverage area of $ed_i$ is denoted by $\mathcal{N}_i(t) = \{n_{i,1}, n_{i,2}, \dots,  n_{i,N_i(t)}\}$, where the cardinality $N_i(t)$ varies over time as a result of user mobility.  

        \item \textit{Feedback Mechanism}: \glspl{mu} report their perceived \gls{SCQ} to \glspl{cp}, enabling the evaluation of caching reliability and quality.
	\end{itemize}

\end{enumerate}

\subsection{Content Architecture}\label{ss:ContentArch}
Over the time horizon $\{1, 2, \dots, T\}$, \glspl{mu} generate requests for different content items. The complete set of available contents, belonging to $cp_o$, within this interval is denoted by $\mathcal{M}_{o} = \{ct^o_1, ct^o_2, \ldots, ct^o_{M_o}\}$. 
Each content item in $\mathcal{M}_{o}$ may be characterized by attributes such as popularity~\cite{zyrianoff2024cache}, importance~\cite{zhang2024cache}, and request distribution~\cite{yang2018content}, defined as follows:

\begin{itemize}[leftmargin=5mm]
	\item \textit{Popularity Distribution}:
    The popularity of the contents is modeled through a probability distribution vector $\mathbf{f}^{o} = [f_{1}^{o}, \dots, f_{M_o}^{o}]$, where each entry $f_{m}^{o}$ specifies the likelihood that an \gls{mu} requests content $ct^o_m$. This distribution is derived by ranking contents in decreasing order of their request frequencies observed during a given time window (e.g., one day or one week). Accordingly, the popularity of $ct^o_m$ is expressed as $f_{m}^o$ and is formally given by (as in \cite{Xu2020})
    \vspace{-3mm}
\begin{equation}\label{eq:Popularity}
    f_{m}^o=
    \left(
    (\tau (m))^{\gamma } \sum _{m =1}^{M_{o}} m^{-\gamma }
    \right)^{-1}.
\end{equation}
\begin{itemize}
\item $\tau(m)$ denotes the index of $ct^o_m$ in the ordering of all contents arranged by request frequency in descending order. According to~\eqref{eq:Popularity}, a smaller index (i.e., higher ranking in terms of requests) corresponds to greater content popularity.

\item $\gamma \geq 0$ is the parameter governing the skewness of the popularity distribution. For $\gamma = 0$, popularity is uniformly distributed across all contents. As $\gamma$ increases, the distribution becomes increasingly skewed, with a small subset of highly popular contents accounting for the majority of requests.
\end{itemize}

\item \textit{Importance Distribution}:
In addition to popularity, the importance of content must be evaluated, as different items may vary in significance. 
The importance is represented by the vector $\mathbf{p}^{o} = [p_{1}^{o}, \dots, p_{M_o}^{o}]$, where $p_{m}^{o}$ denotes the importance assigned to $ct^o_m$. Formally, the importance $p_m^o$ (as in \cite{Xu2020}) is defined as
\begin{equation} 
p_{m}^{o} = 
\left(
(\kappa (m))^{\beta } \sum _{m=1}^{M_o} m^{-\beta }
\right)^{-1}.
\end{equation}
\begin{itemize}
\item $\kappa(m)$ denotes the index of $ct^o_m$ in the descending priority order of all contents in $CT_\mathcal{M}^{o}$. A smaller index corresponds to greater importance.  

\item $\beta \geq 0$ is a parameter that controls the skewness of the importance distribution. Larger values of $\beta$ result in a distribution where a small subset of high-priority contents accounts for the majority of the overall importance.
\end{itemize}

\item \textit{Request Distribution}:     
Content demand differs across \glspl{ed} owing to variations in user preferences within their respective coverage areas. For an $ed_i \in \mathcal{I}$ the request distribution is represented by $\mathbf{r}_{i,o} = [r_{i,1}^{o},  \dots, r_{i,M_o}^{o}]$. 
Here, $r_{i,m}^{o}$ denotes the ratio of \glspl{mu} requesting $ct_m^o$ from $ed_i$. Accordingly, contents with frequent requests from a large number of \glspl{mu} within the coverage area of an \gls{ed} are prioritized for caching on that device.
\end{itemize}

Edge caching devices may behave selfishly or be exposed to open-access vulnerabilities. Rather than modeling specific attack mechanisms, we abstract these effects through an SCQ metric that captures the reliability of cached content.

\vspace{-1mm}
\section{Problem Formulation}\label{sec:problem_formulation}
Within the network setting, when \glspl{ed} deliver high-quality secure caching services, the \gls{QoE} of \glspl{mu} improves, which in turn benefits the \glspl{cp} by attracting more users to access their content, thereby increasing their revenue. Hence, both \glspl{cp} and \glspl{mu} share a common interest in ensuring secure and reliable caching at the edge. To incentivize \glspl{ed} and discourage selfish or malicious behavior, \glspl{cp} adopt payment mechanisms that compensate \glspl{ed} based on the provided \gls{SCQ}. Each \gls{ed}, in response, selects its level of \gls{SCQ} service so as to maximize its own profit, subject to operational and security costs.  
The coexistence of multiple \glspl{cp} gives rise to competition, as all \glspl{cp} seek access to the limited caching resources of \glspl{ed}. Each \gls{cp} aims to secure reliable and high-quality caching services at minimum cost, while \glspl{ed} pursue maximum profit by adjusting their service quality in response to payments. This interaction creates a hierarchical and competitive decision-making process: (i) competition among \glspl{cp} for edge caching resources, and (ii) leader–follower interactions between each \gls{cp} and the \glspl{ed}.  



\subsection{CP Profit Model}

The utility function of each $cp_o$ is denoted as $U_{o}(\mathbf{g}^o, \mathbf{q}^o)$, where $\mathbf{g}^o$ represents the payment strategy of $cp_o$ to all \glspl{ed} for the contents in $\mathcal{M}_o$, and $\mathbf{q}^o$ denotes the corresponding \gls{SCQ} levels provided by \glspl{ed}. Each \gls{cp} adopts a non-uniform payment policy, assigning different payment values to different contents across different \glspl{ed}. Formally, the payment strategy $\mathbf{g}^o$ is expressed as
\begin{equation*}
	\mathbf{g}^o
	\!=\!
	\left [ \begin{array}{llll} 
    g_{1,1}^o,& g_{1,2}^o,& \cdots,& g_{1,M_o}^o\\ 
    g_{2,1}^o,& g_{2,2}^o,& \cdots,& g_{2,M_o}^o\\ 
    \; \vdots & \vdots & \ddots & \vdots \;\\ 
    g_{I,1}^o,& g_{I,2}^o,& \cdots,& g_{I,M_o}^o 
    \end{array} \right]
    \!=\!
    [\mathbf{g}_{1}^o, {\mathbf{g}}_{2}^o, \cdots, \mathbf{g}^o_{I}]^{\textbf {T}},
\end{equation*}
\noindent
where $g_{i,m}^o$ denotes the payment made to $ed_i$ by $cp_o$ to provide secure caching of $ct_m$. 
Accordingly, the matrix capturing the \gls{SCQ} levels offered by all \glspl{ed} is given by
\begin{equation*} 
	\mathbf{q}^o 
	\!=\! 
	\left [
    \begin{array}{lllll}
    {q_{1,1}^o},& {q_{1,2}^o},& \cdots,& {q_{1,M_o}^o}\\ 
    {q_{2,1}^o},& {q_{2,2}^o},& \cdots,& {q_{2,M_o}^o}\\ 
    \vdots & \vdots & \ddots & \vdots \;\\
    {q_{I,1}^o},& {q_{I,2}^o},& \cdots,& {q_{I,M_o}^o}\; 
    \end{array} 
    \right]
    \!=\!
    [\mathbf{q}^o_{1}, \mathbf{q}^o_{2}, \cdots, \mathbf{q}^o_{I}]^{\textbf {T}}.
\end{equation*}
As each \gls{cp} distributes content across multiple \glspl{ed}, its overall utility is obtained by aggregating the individual utilities associated with each content item cached on a given \gls{ed}. 
Hence, the total utility of each $cp_o$ is formulated as
\begin{equation*} 
	U_{o}({\mathbf{g}^o}, {\mathbf{q}^o}) = \sum \limits _{i = 1}^{I} {\sum \limits _{m = 1}^{M_o} {u_{o}({g_{i,m}},{q_{i,m}})} },
\end{equation*}
\noindent
where ${u_{o}(g_{i,m}^o, q_{i,m}^o)}$ denotes the utility gained by the $cp_o$ from caching content $ct^o_m$ on $ed_i$. 
Since the $cp_o$'s utility is determined by the benefit derived from the secure caching service minus the corresponding payment, as initially proposed in~\cite{Xu2020}, it can be expressed as
\begin{equation*} 
	u_{o}(g_{i,m}^o,q_{i,m}^o) = F_{i,m}^o(q_{i,m}^o) - C_{i,m}^o(g_{i,m}^o,q_{i,m}^o).
\end{equation*}
$F_{i,m}^o(q_{i,m}^o)$ represents the satisfaction function of $cp_o$ for caching $ct_m$ on $ed_i$ with a \gls{SCQ} level $q_{i,m}^o$. 
Conversely, $C_{i,m}^o(g_{i,m}^o, q_{i,m}^o)$ captures the $cp_o$'s cost incurred for secure caching of $ct_m$ at $ed_i$. 
Following common practice in resource allocation studies~\cite{zhang2017computing}, the satisfaction function is modeled logarithmically, and is given by
\begin{align}\label{eq:satisfaction_func}
\footnotesize
F_{i,m}^o(q_{i,m}^o) = 
\begin{cases} 
\alpha r_{i,m}^o N_{i}(t) f_{m}^o p_{m}^o \log (1 + q_{i,m}^o), & \text{if $q_{i,m}^o \in [0,1]$},
\\ 
\varsigma r_{i,m}^o N_{i}(t) f_{m}^o p_{m}^o q_{i,m}^o, & \text{if $q_{i,m}^o = -1$}.
\end{cases}
\end{align}
Here, $\alpha > 0$ denotes the satisfaction parameter for secure content caching, while $\varsigma > 0$ reflects the penalty parameter capturing the loss of satisfaction. 
As expressed in~\eqref{eq:satisfaction_func}, the satisfaction function operates piecewise: for $q_{i,m}^o > 0$, the $cp_o$ obtains positive satisfaction from the provided \gls{SCQ}. For $q_{i,m}^o = -1$, the $cp_o$ is deceived by $ed_i$, resulting in negative satisfaction.  
In addition, as $cp_o$ must compensate \glspl{ed} for secure caching services, the corresponding cost function $C_{i,m}^o(g_{i,m}^o, q_{i,m}^o)$ is formulated as:
\begin{align}\label{eq:cost_func} 
C_{i,m}^o(g_{i,m}^o,q_{i,m}^o) = 
\begin{cases} 
g_{i,m}^o \theta q_{i,m}^o, & \text{if $q_{i,m}^o \in [{0,1}]$},
\\ 
0, & \text{if $q_{i,m}^o = -1$}.
\end{cases}
\end{align}
where $g_{i,m}^o$ denotes the payment by $cp_o$ to $ed_i$ for delivering the highest-quality caching service ($q_{i,m}^o=1$) for $ct_m$, while $\theta$ is the payment adjustment parameter; \eqref{eq:cost_func} ensures that an \gls{ed} receives no compensation if it abstains from caching or cheats. Accordingly, by aggregating over all contents and \glspl{ed}, the overall utility of the $cp_o$ at time slot $t$ is expressed as:
{ \small
\begin{align}\label{eq:cp_utility}
U_{o}({\mathbf{g}^o}, {\mathbf{q}^o}) =
\sum \limits _{i = 1}^{I} \sum \limits _{m = 1}^{M_o}  
{x_{i,m}^o} \left[ 
\alpha {r_{i,m}^o}{N_{i}(t)}{f_{m}^o}{p_{m}^o}\log (1 + {q_{i,m}^o}) \right.
\nonumber
\\
 \left. - {g_{i,m}^o}\theta {q_{i,m}^o} \right] 
+\, (1 - {x_{i,m}^o})\varsigma {r_{i,m}^o}{N_{i}(t)}{f_{m}^o}{p_{m}^o}{q_{i,m}^o}
\end{align}
}
where $x_{i,m}^o \in \{0,1\}$ is defined as the service integrity indicator associated with $cp_o$, $ed_i$, and $ct_m$:  
\begin{equation*} 
x_{i,m}^o= 
\begin{cases} 
1, & \text{if $q_{i,m}^o \in [{0,1}]$},
\\ 
0, & \text{if $q_{i,m}^o = -1$}. 
\end{cases}
\end{equation*}
Given limited budgets and varying content importance, each \gls{cp} allocates a global budget $G_m^o$ to content $ct_m$, representing the maximum total caching cost across all \glspl{ed}. For each content, the \gls{cp} also sets lower and upper per-device payment bounds $G_m^{o,L}$ and $G_m^{o,U}$, with $G_m^{o,U} \leq G_m^o$. Under these constraints, the profit maximization problem for $cp_o$ and content $ct_m$ is formulated as follows:
\begin{equation*}\label{eq:platformOBJ}
	\max_{\mathbf{g_m^o}}\ 
	U_{o}({\mathbf{g}^o}, {\mathbf{q}^o}) =
	\sum_{i=1}^I {F_{i,m}^o}({q_{i,m}^{o*}})  - g_{i,m}^o \theta \, q_{i,m}^{o*}(\mathbf{g_m^o})
\end{equation*}
subject to:
{\small
\begin{align}
	& q_{i,m}^{o*}(\mathbf{g}_m^o) = 
	\arg\max_{\mathbf{q}_{i}^o} 
	\{
	x_{i,m}^o
	(g_{i,m}^o \theta q_{i,m}^o - 
    c_{i} \nu (q_{i,m}^m)^{2}) 
    \nonumber &
    \\
	& 
    \qquad - (1 - {x_{i,m}^o}) \psi _{i} : q_{i,m}^o \in [0,1] \cup \{-1\}
	\},
 \hspace{1cm} \forall i \in \mathcal{I},
 \label{eq:EDoptresp}
	\\
	&
	G_{m}^{o,L} \leq g_{i,m}^o \leq G_m^{o,U} 
    \hspace{4cm}
	\forall i \in \mathcal{I}, 
	\label{c:priceBound}
	\\
	& 
	\sum_{i=1}^I g_{i,m}^o \leq G_m^o. 
	\label{c:budgetConstraint}
\end{align}
}
\noindent Here, $q_{i,m}^{o*}(\mathbf{g_m^o})$ is the optimal strategy of $ed_i$ from~\eqref{eq:best_reply}. 



\subsection{ED Profit Model}
\label{ss:ed_problem}

Each \gls{ed} evaluates its utility as the net benefit obtained from secure caching, defined as the difference between the payment received from \glspl{cp} and the corresponding service cost. 
Thus, the utility function of $ed_i$, computed as in~\cite{Xu2020}, is given by
\begin{equation} 
	U_{i}(\mathbf{q}_i, \mathbf{g}_i) = L_{i}(\mathbf{q}_i, \mathbf{g}_i) - \Phi _{i}(\mathbf{q}_i),
\end{equation}
where $\mathbf{g_i} = (\mathbf{g}_i^1, \dots, \mathbf{g}_i^O)$ denotes the vector of payment strategies of all \glspl{cp} to $ed_i$, and $\mathbf{q}_i = (\mathbf{q}_i^1, \dots , \mathbf{q}_i^O)$ represents the vector of \gls{SCQ} that $ed_i$ provides to all \glspl{cp}.

Here, $L_{i}(\mathbf{q}_{i}, \mathbf{g}_i)$ represents the payment received from the all \glspl{cp}, determined by the payment strategy $\mathbf{g}_i^o$ and the vector of \gls{SCQ} levels $\mathbf{q}_i^o$ across the $M_o$ contents, i.e.,
\begin{equation*} 
L_{i}(\mathbf{q}_{i}, \mathbf{g}_i) 
= 
\sum \limits_{o = 1}^{O} \sum \limits_{m=1}^{M_o} C_{i,m}^{o}({g_{i,m}^{o}},q_{i,m}^{o}).
\end{equation*}
The function $\Phi_{i}(\mathbf{q}_i)$ characterizes the service cost of $ed_i$ when providing secure caching to all \glspl{cp}. 
As higher \gls{SCQ} levels demand additional computational, storage, and security resources, the cost naturally increases with service quality. Accordingly, the cost function of $ed_i$ is given by
\begin{equation*}
	\Phi _{i}(\mathbf{q}_{i}) = \sum \limits _{o = 1}^{O} \sum \limits _{m=1}^{M_o} \varphi _{i,m}^{o}({q_{i,m}^{o}}),
\end{equation*}
where $\varphi_{i,m}^{o}(q_{i,m}^{o})$ denotes the cost incurred by $ed_i$ for a secure caching service of quality $q_{i,m}^{o}$ on $ct_m$ for $cp_o$, and is expressed as
\begin{equation} \label{eq:ed_cost} 
	{\varphi _{i,m}^{o}}({q_{i,m}^{o}}) = \begin{cases} {c_{i}}\nu (q_{i,m}^{o})^2,
    & 
    \text{if } q_{i,m}^{o} \in [{0,1}],
    \\ 
    {\psi _{i}},
    &
    \text{if } q_{i,m}^{o} = - 1. 
    \end{cases}
\end{equation}
The parameter $c_i$ represents the cost incurred by $ed_i$ in delivering the highest level of \gls{SCQ}, capturing its total expenditure for this service \cite{Xu2020}. The term $\nu$ is an adjustment parameter specific to $ed_i$, while $\psi_i$ denotes a fixed value reflecting the resource consumption of $ed_i$ (e.g., power, bandwidth) when engaging in cheating behavior against \glspl{cp}~\cite{Xu2020, liu2017incentive, xu2017secure}. 
Accordingly, the utility function of $ed_i$ can be expressed as
\begin{align*} 
	U_{i}(\mathbf{q}_i, \mathbf{g}_i)
	=\sum \limits _{o = 1}^{O} \sum \limits _{m=1 }^{M_o} u_{i,m}^{o}(q_{i,m}^{o},g_{i,m}^{o}),
\end{align*}
where ${u_{i,m}^{o}}({q_{i,m}^{o}},{g_{i,m}^{o}})$ is the utility of $ed_i$ to securely cache $ct^o_m$. Here, we have
{\small
\begin{equation*} 
	u_{i,m}^{o}(q_{i,m}^{o},g_{i,m}^{o}) = 
	\begin{cases} 
		g_{i,m}^{o}\theta q_{i,m}^{o} - c_{i}\nu (q_{i,m}^{o})^2, 
        & 
        \text{if } q_{i,m}^{o} \in [0,1],
        \\
		- \psi_i, 
        &
        \text{if } q_{i,m}^{o} = - 1. 
	\end{cases}
\end{equation*}
}
From Eq. \eqref{eq:cost_func} and Eq. \eqref{eq:ed_cost}, $C_{i,m}^o(0,-1)=C_{i,m}^o(0,0)=0$, whereas  $\varphi _{i,m}^{o}(-1)>\varphi _{i,m}^{o}(0)=0$. This implies that, when evaluating the quality of the secure caching service provided by each edge device, a device engaging in cheating behavior attains a lower utility than one that does not participate in secure content caching. Specifically, $u_{i,m}^{o}(0, -1)=- \psi_i< u_{i,m}^{o}(0, 0)$. Hence, a zero-payment penalty effectively discourages selfish behavior among nodes. Consequently, the optimization problem for maximizing the utility of $ed_i$ can be written as:
{\small
\begin{align}\label{eq:ed_problem}
	\max\limits_{q_i, y_i} \ 
    \sum\limits_{o = 1}^{O} \sum \limits_{m =1}^{M_o} 
    \left[
        {g_{i,m}^{o}}\theta {q_{i,m}^{o}} - {c_{i}}\nu (q_{i,m}^{o})^2
    \right]
\end{align}
}
\noindent subject to:
{\small
\begin{alignat}{1}
&\sum_{o=1}^{O} \sum_{m=1}^{M_o} \delta_{i,m}^{o} y_{i,m}^{o} \leq D_i,
\label{c:storageLimit} 
\\
&q_{i,m}^{o} \leq y_{i,m}^{o}, \qquad  \forall o=1,\dots,O,\ m=1,\dots,M_o,
\label{c:q_i_y_i} 
\\
&q_{i,m}^{o} \geq \epsilon y_{i,m}^{o}, \qquad \forall o=1,\dots,O,\ m=1,\dots,M_o,
\label{c:q_i_y_i_eps} 
\\
&y_{i,m}^{o} \in \{0,1\}, \qquad \forall o=1,\dots,O,\ m=1,\dots,M_o.
\label{c:ybinary}
\end{alignat}
}
Here, $q_{i,m}^{o}$ is a continuous decision variable representing the secure caching quality selected by $ed_i$ for content $ct^o_m$ of $cp_o$, while $y_{i,m}^{o} \in \{0,1\}$ is a binary variable indicating whether $ed_i$ caches $ct^o_m$ of $cp_o$ ($y_{i,m}^{o}=1$) or not ($y_{i,m}^{o}=0$). 
The parameter $\delta_{i,m}^{o}$ denotes the storage size required to cache content $ct^o_m$ of $cp_o$ at $ed_i$, and constraint~\eqref{c:storageLimit} ensures the total cached content does not exceed $ed_i$ capacity $D_i$. Constraint~\eqref{c:q_i_y_i} further ensures that positive values of $q_{i,m}^{o}$ occur only if the content is cached, with a small positive constant $\epsilon$ guaranteeing that caching provides a quality greater than zero.

\vspace{-1mm}
\section{Game-Based Optimal Strategy Analysis}\label{sec:Stackelberg_game}

In this section, we analyze the proposed game-theoretic framework under two distinct storage regimes that arise in practical edge caching systems. We first consider a light storage constraint regime in Section \ref{sec:light_constraint}, where ED capacity is sufficiently provisioned such that storage constraints are non-binding. This regime enables a rigorous game-theoretic characterization and decentralized convergence guarantees. We then examine a strict storage constraint regime in Section \ref{sec:strict_constraint}, in which ED capacity becomes binding due to high demand or limited resources, leading to coupled CP decisions and the loss of potential-game structure.


\subsection{Optimal Strategy under Light Storage Constraints}
\label{sec:light_constraint}

In this subsection, we assume that ED storage capacity is sufficient such that constraints \eqref{c:storageLimit} are inactive and do not restrict caching decisions.
In the considered system, each \gls{cp} chooses payments $g_{i,m}^o$ for caching $ct_m^o\in \mathcal{M}_o$ at available \glspl{ed}, aiming to maximize its net utility under content-level budgets. Each \gls{ed}, in turn, selects the secure-caching quality $q_{i,m}^o$ for cached content items to maximize its own profit subject to storage limits. The resulting interaction is hierarchical and competitive: \glspl{cp} act as leaders that announce payment strategies, \glspl{ed} act as followers that respond with quality levels, and \glspl{cp} compete among themselves for the limited caching capacity of \glspl{ed}.

To capture these interactions, we adopt a multi-leader-multi-follower game-theoretic framework. 
The interaction between each \gls{cp} (as leader) and the \glspl{ed} (as followers) is captured by a single-leader-multi-follower Stackelberg game serving as the condition that determines the \glspl{ed}’ responses. 
The model captures both the leader–follower dynamics of \gls{cp}–\gls{ed} interactions and the competitive behavior of multiple \glspl{cp}. 
We show that under a light storage constraint, this CP competition constitutes an exact potential game, which guarantees both the existence of a pure-strategy Nash equilibrium (NE) and the convergence of a decentralized iterative algorithm (\autoref{alg:BestResponse}) to such equilibria.

\textbf{Optimal response of EDs}\label{sec:ed_response}:
Given a payment $g_{i,m}^o$ offered by the $cp_o$, the $ed_i$  solves its local profit problem as defined in~\eqref{eq:ed_problem}--\eqref{c:ybinary}. 
As in \cite{Xu2020} and  \cite{singleProvider}, the optimal strategy and best response of $ed_i$ on $ct^o_m$ is:
\begin{equation} \label{eq:best_reply}
q_{i,m}^{o*} = 
\begin{cases} 
1, & \text{if } g_{i,m}^o \geq 2{c_{i}}\nu/\theta,
\\[2mm]
\dfrac {g_{i,m}^o \theta }{{2{c_{i}}\nu }},
& \text{if } 0 \leq {g_{i,m}^o} \leq 
2 c_{i} \nu / \theta.
\end{cases}
\end{equation} 

Thus, an \gls{ed} provides full quality when the offered payment exceeds its saturation payment ($2 c_i \nu/\theta$), a proportional quality otherwise, and no service if no payment is received.

\textbf{Optimal response of CPs}
\label{sec:cp_response}:
With the \glspl{ed}' best-response $q_{i,m}^{o*}$ obtained by solving~\eqref{eq:ed_problem}--\eqref{c:ybinary}, we now focus on the interaction among \glspl{cp}. 
Each \gls{cp} determines its payment strategy $\mathbf{g}^o$ to maximize its own utility.
By substituting the EDs’ best responses into the CP utilities, the interaction among CPs can be analyzed independently at the leader level.
Formally, this game can be described as follows:

\begin{definition}\label{def:CP_game}
The strategic game between \glspl{cp} is the triplet $\mathcal{G}_{\text{CP}} = \big(\mathcal{O}, \{ S_o \}_{cp_o \in \mathcal{O}}, \{ U_o \}_{cp_o \in \mathcal{O}} \big)$, where:
\begin{itemize}
    \item $\mathcal{O}$ is the set of all \glspl{cp}.
    
    \item $S_o$ is the strategy set of $cp_o$. For each $cp_o \in \mathcal{O}$, the strategy is its payment vector $\mathbf{g}^o$, where $\sum_{i=1}^I g_{i,m}^o \le G_m^o, \ \forall m \in M_o$ and $G_{m}^{o,L} \leq g_{i,m}^o \leq G_m^{o,U}, 
    \
	\forall i \in \mathcal{I}$
    \item $U_o$ is the $cp_o$'s utility function as in~\eqref{eq:cp_utility}.
\end{itemize}
\end{definition}


In the following, we show that the game $\mathcal{G}_{\text{CP}}$ constitutes an \emph{exact potential game}, which guarantees the existence of a pure-strategy Nash equilibrium and enables decentralized computation via best-response dynamics. 
Specifically, $\mathcal{G}_{\text{CP}}$ is an \emph{exact potential game} if there exists a potential function $P: \prod_{o=1}^O S_o \rightarrow \mathbb{R}$ such that, for each $cp_o \in \mathcal{O}$, any strategy profile 
 $(\mathbf{g}^o,\mathbf{g}^{-o})$, and any unilateral deviation 
 $\tilde{\mathbf{g}}^o \in S_o$, the following holds:
 \begin{equation*}
 U_o(\tilde{\mathbf{g}}^o,\mathbf{g}^{-o}) - U_o(\mathbf{g}^o,\mathbf{g}^{-o})
 \;=\; P(\tilde{\mathbf{g}}^o,\mathbf{g}^{-o}) - P(\mathbf{g}^o,\mathbf{g}^{-o}).    
 \end{equation*}

\begin{theorem}
\label{theorem:potential}
The game $\mathcal{G}_{\text{CP}}$ is an exact potential game with potential function
{\scriptsize
\begin{align}
\label{eq:potential_func}
\begin{array}{r}
\displaystyle 
P(\mathbf{g}) = 
\sum\limits_{o=1}^O \sum\limits_{i=1}^I \sum\limits_{m=1}^{M_o} 
\left[ 
\alpha {r_{i,m}^o}{N_{i}(t)}{f_{m}^o}{p_{m}^o} \log\left( 1 + \frac{g_{i,m}^{o} \theta}{2 c_i \nu}\right) 
- \frac{(g_{i,m}^{o})^2 \theta^2}{2 c_i \nu} \right].
\nonumber
\end{array}
\end{align}
}
\end{theorem}

\begin{proof}
The proof is in \autoref{sec:Appendix}.
\end{proof}

Since the utility function of each CP is independent of the other CPs' strategies, Nash equilibria coincides with the optimal strategies of the potential function. Moreover, as the potential function $P$ is strictly concave, the considered game admits a unique Nash equilibrium.

\begin{algorithm}
		\caption{Decentralized Payment Update (best reply to $q_i$)}
		\label{alg:BestResponse}
		{\footnotesize
			\begin{algorithmic}[1]
				\State{\textbf{Input:} $\mathcal{O}= \{cp_1, .., cp_O\}$, $ED_{\mathcal{I}}= \{ed_1, .., ed_I\}$, budget limits $G_m^o$} 
				\State{\textbf{Output:} Equilibrium payment strategies $g^* = (g_o)_{o=1}^O$, optimal budget allocations}
				\State{\textbf{Step 1: CP Proposes initial Payments to EDs}}
				\For{$cp_o \in \mathcal{O}$}
				\State{Compute initial budget allocation across EDs:}
				\[
				g_{i,m}^{o} = \frac{G_m^o}{I} \quad \forall ed_i \in ED_{\mathcal{I}}
				\]
				\State{Propose payments $g_{i,m}^{o}$ to each ED (initial)}
				\EndFor\State{\textbf{Step 2: CP Payment Proposal and ED Response}}
				\While{At least one $cp_o$ has an incentive to update}
				\State{\textbf{Step 2.1: ED Response to \glspl{cp}}}
				\For{$ed_i \in ED_{\mathcal{I}}$}
				\State{$ed_i$ received payment $g_{i,m}^{o}$ from $cp_o$ for each CP and provides the response as secure caching quality $q_{i,m}^o$}
				\EndFor
				\State{\textbf{Step 2.2: CP Payoff Calculation and Budget Adjustment}}
				\For{$cp_o \in CP_{\mathcal{O}}$}
				\State{Compute payoff $U_o$ based on EDs response $q_{i,m}^o$ and $g_{i,m}^{o}$}
				\State{Update budget allocation based on $q_{i}^o$ (ES Response):}
				\[ \textbf{g}_{m}^{o}= \arg \max_{g_m^o} U_o \quad \text{s.t.} \quad \sum_{i} g_{i,m}^o \leq G_m^o \]
				\EndFor
				
				\If{No CP improves its payoff}
				\State{\textbf{Convergence reached, stop updating.}}
				\EndIf
				\EndWhile
				
				\State{\textbf{return} Equilibrium payments $g^*$ and optimal budget allocation}
			\end{algorithmic}
		}
	\end{algorithm}

Building on this property, \autoref{alg:BestResponse} outlines a decentralized iterative procedure that enables each \gls{cp} to converge toward its equilibrium payment strategy. The algorithm begins by initializing the system with the set of \glspl{cp}, \glspl{ed}, and the budget constraints $G_m^o$ (Lines~1-2). Each \gls{cp} then initializes its payment strategy $g_{i,m}^o$ for each $ed_i$ and content $ct^o_m$, typically by distributing its budget $G_m^o$ equally among the \glspl{ed}. These initial proposals are sent to the \glspl{ed} as the first round of interactions (Lines~3-7). In the main iterative phase (Lines~8 onward), each \gls{ed} computes its best-response caching quality $q_{i,m}^o$ for the received payments using the closed-form expression in~\eqref{eq:best_reply} (Lines~10–13). Given these responses, each \gls{cp} recomputes its utility and updates its budget allocation by solving its local optimization problem subject to budget constraints (Lines~14–18). This represents a best-response update in the strategy space of the \glspl{cp}. Finally, a convergence check is performed (Lines~19–21): if no \gls{cp} can further improve its utility, the procedure terminates and the current profile is returned as the equilibrium solution (Line~22). 
Owing to the potential game structure at the leader level, this decentralized best-response process converges to the pure NE strategy profile.

\subsection{Strategy Dynamics under Strict Storage Constraints} \label{sec:strict_constraint}

When EDs storage constraints~\eqref{c:storageLimit} are strictly binding, CPs decisions become coupled through shared capacity limits, 
and the potential game structure established in Section \ref{sec:light_constraint} no longer applies. 
As a result, closed-form equilibrium characterization and convergence proofs become analytically intractable.


Nevertheless, strict storage constraints are of significant practical interest, as real-world edge systems frequently operate under high demand or limited storage availability. Accordingly, rather than enforcing restrictive assumptions to recover theoretical guarantees, we investigate this regime through extensive simulation by relying on \autoref{alg:BestResponse}. As demonstrated in Section 6, the proposed decentralized protocol consistently converges to stable strategy profiles across a wide range of system configurations, indicating robust empirical behavior beyond the analytically tractable regime.


\vspace{-1mm}
\section{Performance Evaluation}\label{sec:experimental}
This section presents numerical analysis results evaluating the proposed decentralized edge caching framework performance. We first outline the experimental setup in Section \ref{sec:exp_setup}, then analyze numerical results across key performance indicators including entity utilities, convergence behavior, and execution time with respect to system parameters, in Section \ref{sec:exp_Numerical_Results}. 
Following the theoretical analysis in~\autoref{sec:Stackelberg_game}, we conduct experiments under two distinct storage regimes: light storage constraints and strict storage constraints. All analyses were conducted on a MacBook Air with an M1 chip, featuring an 8-core CPU at 3.2 GHz and 16 GB of RAM, using BARON 24 for optimization tasks with default settings.

\subsection{Experimental Setup}\label{sec:exp_setup}
\begin{table}[bp]
    \centering
    \small
    \caption{Simulation Parameters}
    \label{tab:sim_params}
    		\scalebox{0.8}{
    \begin{tabular}{l  l | l  l}
    \toprule
    \textbf{Parameter} & \textbf{Value/Range} & \textbf{Parameter} & \textbf{Value/Range} \\ 
    \midrule
        $O$ & 1 to 5 & $\alpha$ & 20 \\ 
        $I$ & 5 to 50 & $\theta$ &  $[0.8, 1.2]$\\
        $M_o$ & 5 to 25 & $\nu$ &  $[0.8, 1.2]$\\ 
        $c_i$ & $[0.5, 1.5]$ & $G_m^{o,L}$ & 0 \\ 
        $N_i(t)$ &  $[50, 100]$ & $G_m^{o,U}$ & $[0.5,5]$ \\ 
        $f_m^{o}$ Distribution & Zipf with $\gamma = 1.2$ & $G_m^{o}$ & $[2, 50]$\\ 
        $p_m^o$ Distribution & Zipf with $\beta = 1.2$ & $\delta_m^o$ & $[5,15]$\\ 
        
    \bottomrule
    \end{tabular}
    }
\end{table}

The game model instances have been generated by adopting the parameters listed in \autoref{tab:sim_params}, based on the values adopted in~\cite{Xu2020}. We consider networks with $1$ to $5$ \glspl{cp} and $5$ to $50$ \glspl{ed}, where the number of \glspl{ed} is varied in steps of $5$. An \gls{ed}'s cost parameter, $c_i$, is drawn uniformly from the interval $[0.5,1.5]$ reflecting varying operational costs across
EDs, to capture heterogeneity in deployment expenses across \glspl{ed}. The number of \glspl{mu} associated with each \gls{ed} is sampled uniformly between $50$ and $100$. Each $cp_o$ has between 5 and 25 contents, with content sizes $\delta_m^o$ drawn uniformly from $[5, 15]$~MB. Both content popularity and content importance are modeled using Zipf distributions with skewness parameters specified in \autoref{tab:sim_params}.
For each $ed_i$, we consider $D_i \in [1.5, 2]$~GB for light storage, and $D_i \in [150, 200]$~MB for strict storage, modeling resource-constrained edge environments. For every fixed configuration, $10$ independent random instances are generated, and all reported results are averaged across these instances to ensure statistical robustness.

\subsection{Experimental results}\label{sec:exp_Numerical_Results}

This section presents numerical results evaluating the proposed algorithm. We first analyze its convergence behavior and scalability under different network sizes and storage constraints (see Section \ref{sec:Convergence_Scalability}). We then study the impact of the \gls{ed} cost parameter on the utilities of \glspl{cp} and \glspl{ed}, highlighting the effects of pricing and storage limitations on system performance (see Section \ref{sec:cost_parameter_impact}).

\subsubsection{Convergence and Scalability} \label{sec:Convergence_Scalability}

\begin{figure}[tbp]
  \captionsetup[subfigure]{justification=centering} 
  \begin{subfigure}[b]{0.22\textwidth}
    \includegraphics[width=\linewidth]{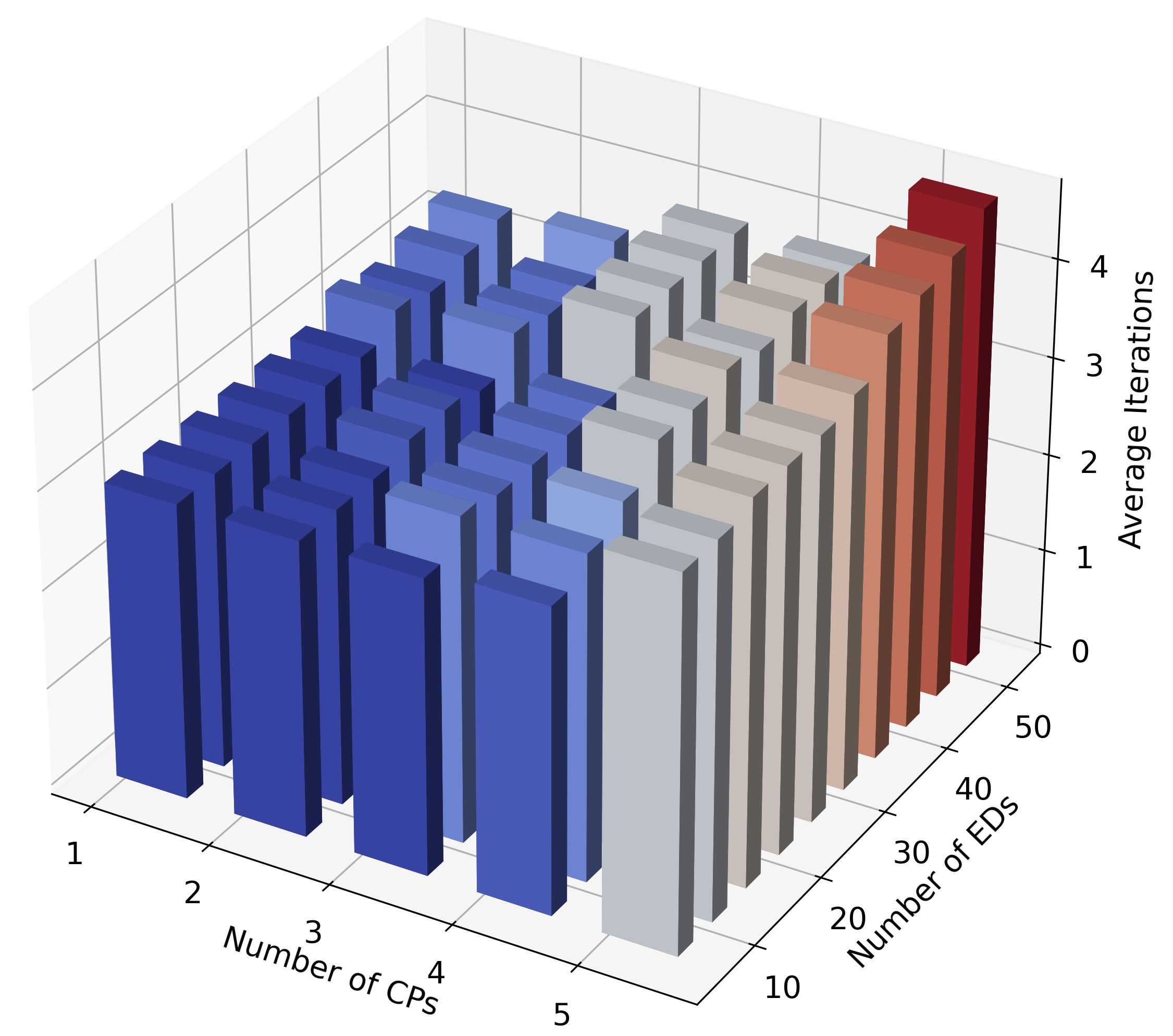}
    \caption{ Iterations.}
    \label{plot:light_iteration}
  \end{subfigure}
  \begin{subfigure}[b]{0.22\textwidth}
    \includegraphics[width=\linewidth]{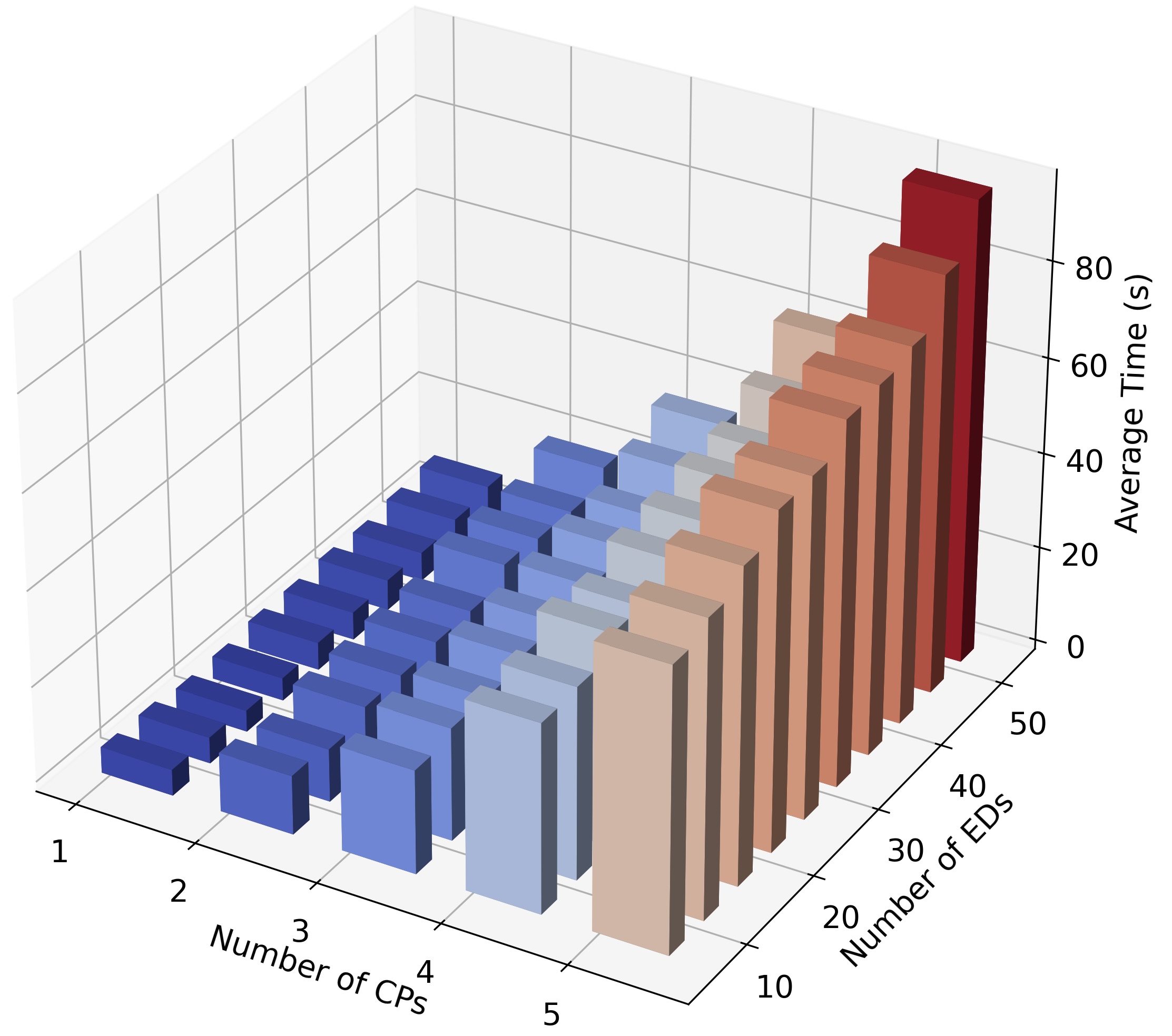}
    \caption{Execution time.}
    \label{plot:light_time}
  \end{subfigure}\hfill

  \caption{Average iterations and execution time to converge vs. number of \glspl{cp} and \glspl{ed}, under light storage constraint.}
  \label{plot:Ave_It_Exec_time_light}
\end{figure}

\begin{figure}[tbp]
  \captionsetup[subfigure]{justification=centering} 
  \begin{subfigure}[b]{0.22\textwidth}
    \includegraphics[width=\linewidth]{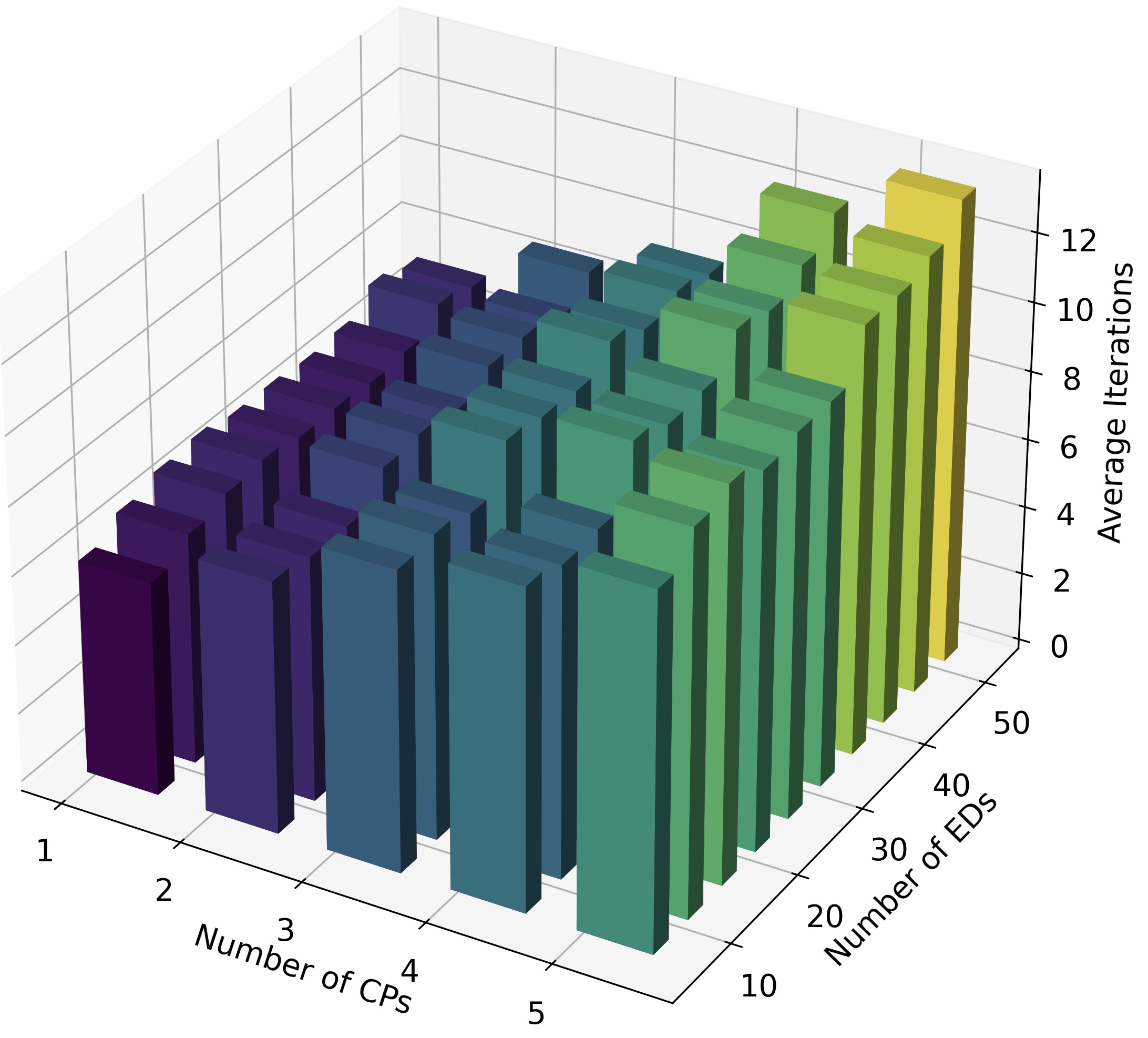}
    \caption{ Iterations.}
    \label{plot:strict_iteration}
  \end{subfigure}
  \begin{subfigure}[b]{0.22\textwidth}
    \includegraphics[width=\linewidth]{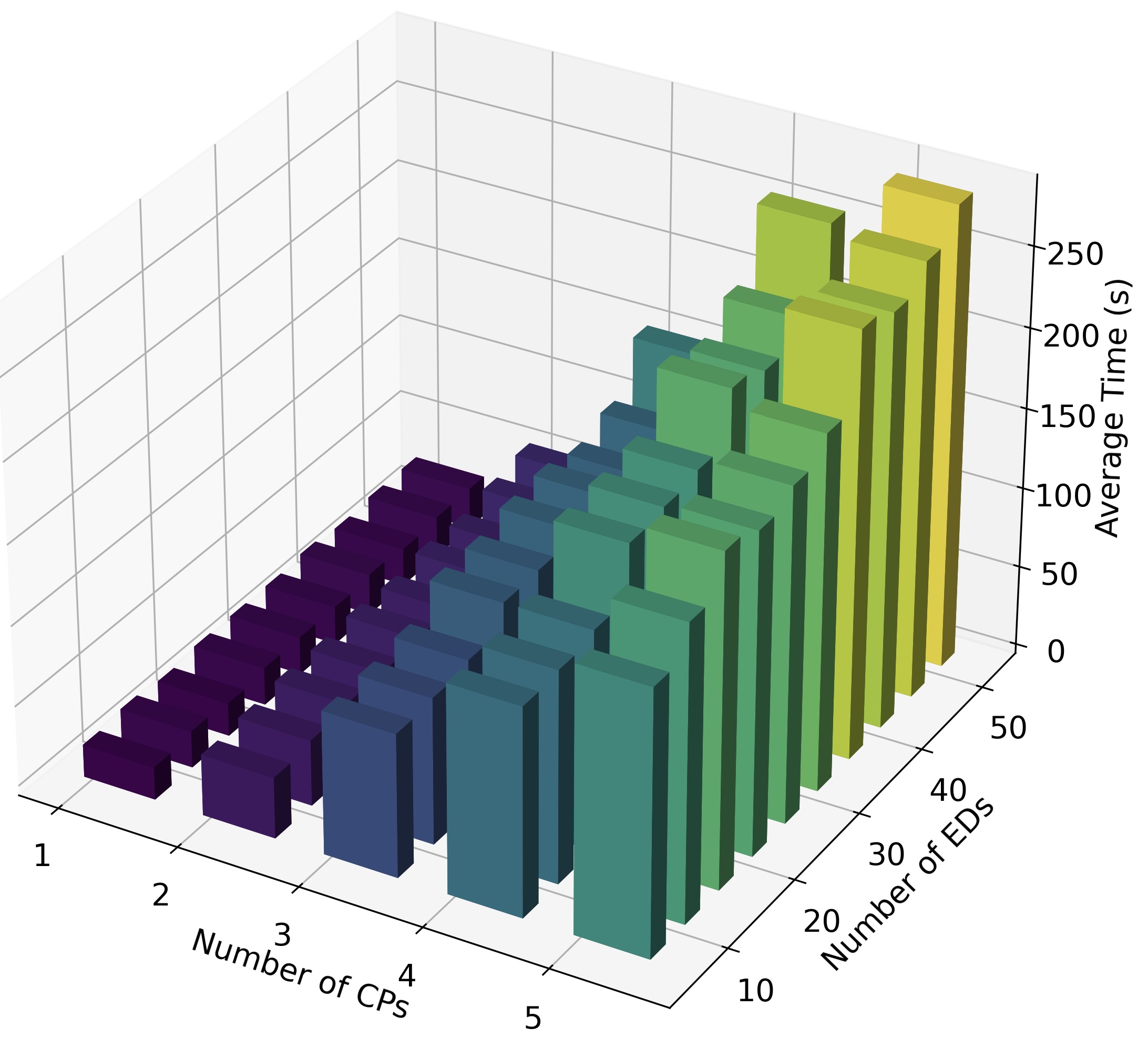}
    \caption{Execution time.}
    \label{plot:strict_time}
  \end{subfigure}\hfill

  \caption{Average iterations and execution time to converge vs. number of \glspl{cp} and \glspl{ed}, under strict storage constraint.}
  \label{plot:Ave_It_Exec_time_strict}
\end{figure}

Figures~\ref{plot:light_iteration}--\ref{plot:strict_time} present the convergence behavior of \autoref{alg:BestResponse}, under both light and strict storage constraints across varying network scales. Figures~\ref{plot:light_iteration} and \ref{plot:light_time} show the average number of iterations and execution time required to reach convergence under light storage constraints, and Figures~\ref{plot:strict_iteration} and \ref{plot:strict_time} show the average number of iterations and execution time required to reach convergence under strict storage constraints. Under light storage ($D_i \in [1.5, 2]$~GB), \glspl{ed} can cache the entire contents from all \glspl{cp} (5-125 content per), while under strict storage ($D_i \in [150, 200]$~MB), \glspl{ed} can only store 10-40 contents, forcing selective caching of approximately 1-3 contents per \gls{cp} depending on content sizes.

Under light storage constraints, the algorithm demonstrates rapid convergence across all configurations. For small-scale networks (1-2 \glspl{cp}, 5-20 \glspl{ed}), convergence is achieved within 3 iterations, taking approximately 5-20 seconds. As the network scales to larger configurations (4-5 \glspl{cp}, 40-50 \glspl{ed}), the iteration increases moderately to 4 iterations, with execution times reaching 80 seconds. The relatively flat iteration profile across varying numbers of \glspl{ed} (for fixed number of \glspl{cp}) indicates that the algorithm's convergence is primarily influenced by the number of \glspl{cp} rather than \glspl{ed}, which aligns with the theoretical analysis showing that coordination complexity grows with the number of participating content providers.
Even with a $5\times$ increase in CPs and $10\times$ increase in EDs, iterations increase only $2\times$, demonstrating robust convergence.

Under strict storage, convergence requires more iterations, especially in large networks. For 1--2 CPs, 7--9 iterations (20--50~s) are needed, roughly double that of light storage. For 5 CPs with 50 EDs, convergence takes 12--13 iterations (260--280~s), reflecting the overhead from frequent invocation of the importance-weighted fallback mechanism due to limited caching. Across both storage regimes, iterations are more sensitive to the number of CPs than EDs: increasing CPs from 1 to 5 raises iterations by 40--50\% (light) and 80--90\% (strict), while increasing EDs from 5 to 50 increases iterations by only 10--20\%. 
Even under strict constraints, all scenarios converge within 5~minutes solving all problems sequentially.  In practice, CPs and EDs problems can be solved in parallel in distributed settings. 
For the largest scale, the solution of the CP problem takes on average 10 s under strict constraints (5 s under light constraints), while the solution of the ED problem takes 5 s (2 s under light constraints), validating the practical feasibility of our approach. These results further validate the practical feasibility of our approach in distributed real-world settings.

\subsubsection{Cost parameter impact on CPs and EDs Utilities.}\label{sec:cost_parameter_impact}

In this analysis, we examine the impact of the cost parameter $c_i$ of an \gls{ed} on its utility across different numbers of \glspl{cp} and the utilities of \glspl{cp} under both light and strict storage constraints. These metrics are evaluated as $c_i$ varies from 0.5 to 1.5, with 5 \glspl{cp} each offering 5 content items, providing a comprehensive understanding of the influence of cost and storage constraints on the system.

As shown in Figure~\ref{plot:c_i_impact_on_cp}, CP utilities decrease monotonically with increasing 
$c_i$ under both storage regimes, reflecting the higher payments required to incentivize EDs as their operational costs rise. 
Under light storage constraints (\autoref{plot:cp_light_c_i}), all CPs maintain positive utilities across the entire cost range, though substantial performance differences emerge due to heterogeneous content popularity, importance, and budget levels. In this regime, abundant storage allows multiple CPs to coexist profitably despite cost increases.

Under strict storage constraints (\autoref{plot:cp_strict_c_i}), CP utilities are significantly reduced, with declines ranging from approximately 50\% to nearly 100\% compared to the light-storage case. Limited storage forces EDs to cache only a small subset of content, intensifying competition among CPs and disproportionately disadvantaging those with lower effective valuations. As a result, some CPs obtain near-zero utility across all values of $c_i$, revealing a scarcity-driven “winner-takes-most” outcome.


\begin{figure}[htbp]
  \captionsetup[subfigure]{justification=centering} 
  \begin{subfigure}[b]{0.24\textwidth}
    \includegraphics[width=\linewidth]{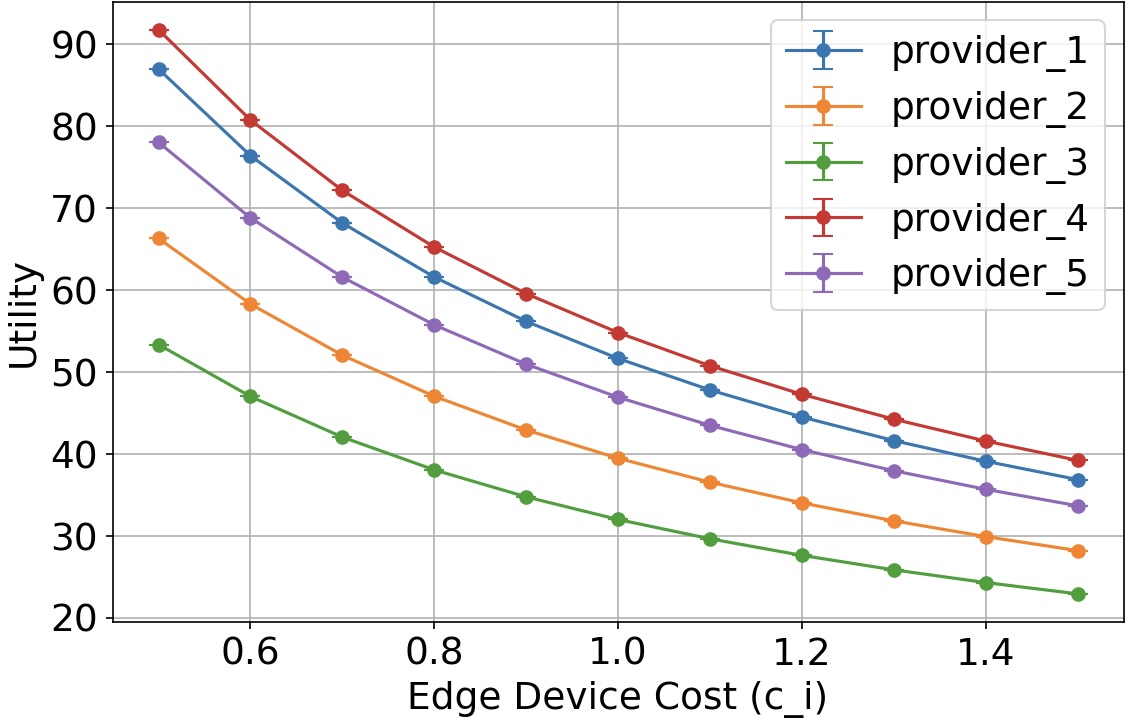}
    \caption{Light storage constraint. }
    \label{plot:cp_light_c_i}
  \end{subfigure}
  \begin{subfigure}[b]{0.24\textwidth}
    \includegraphics[width=\linewidth]{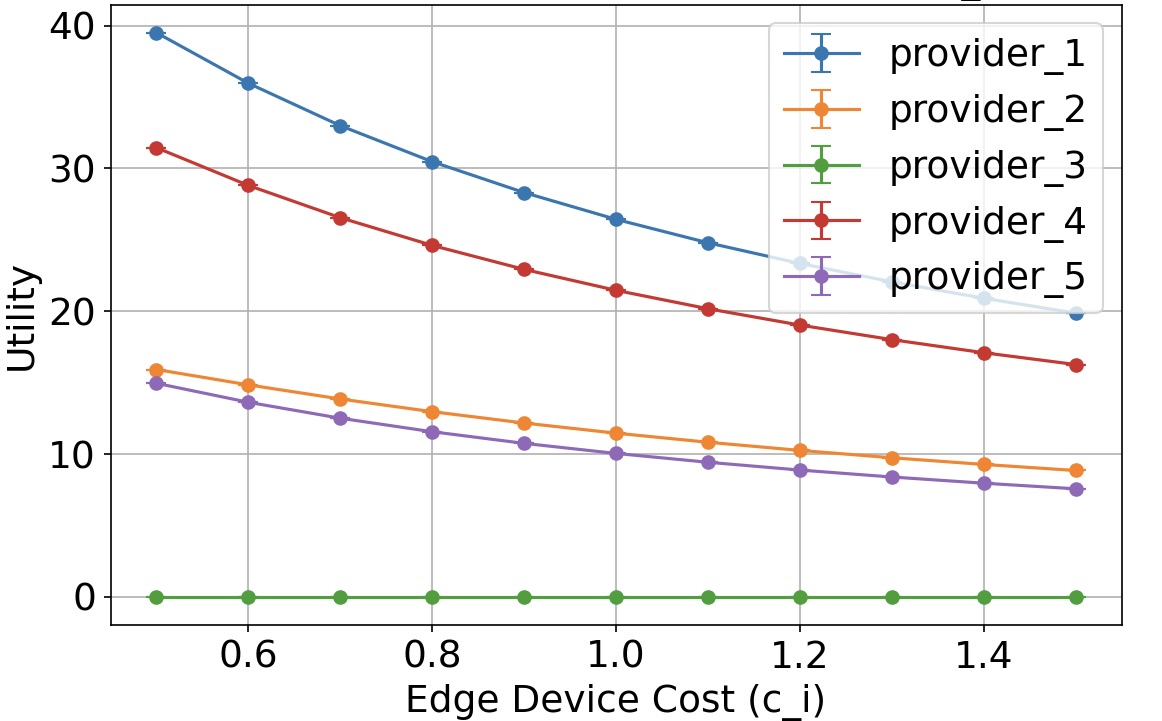}
    \caption{Strict storage constraint.}
    \label{plot:cp_strict_c_i}
  \end{subfigure}\hfill

  \caption{Average \gls{cp} utilities vs. \gls{ed} cost parameter $c_i$.}
  \label{plot:c_i_impact_on_cp}
\end{figure}

\autoref{plot:c_i_impact_on_ed}  illustrates the corresponding impact on ED utilities. While ED utility also decreases as $c_i$ increases, it rises markedly with the number of competing CPs under both storage regimes. Under light storage (\autoref{plot:ed_light_c_i}), EDs benefit substantially from increased CP competition, achieving up to a threefold utility increase when the number of CPs grows from one to five. Under strict storage (\autoref{plot:ed_strict_c_i}), ED utilities are lower overall due to limited caching capacity, but remain less severely impacted than CP utilities, declining by approximately 60–70\%.
Overall, these results indicate that storage scarcity shifts economic surplus away from CPs—particularly weaker ones—and toward EDs, which retain relative bargaining power by controlling access to limited caching resources. While increased CP competition benefits EDs in all cases, strict storage constraints exacerbate inequality among CPs and reduce overall market inclusiveness.

\begin{figure}[htbp]
  \captionsetup[subfigure]{justification=centering} 
  \begin{subfigure}[b]{0.22\textwidth}
    \includegraphics[width=\linewidth]{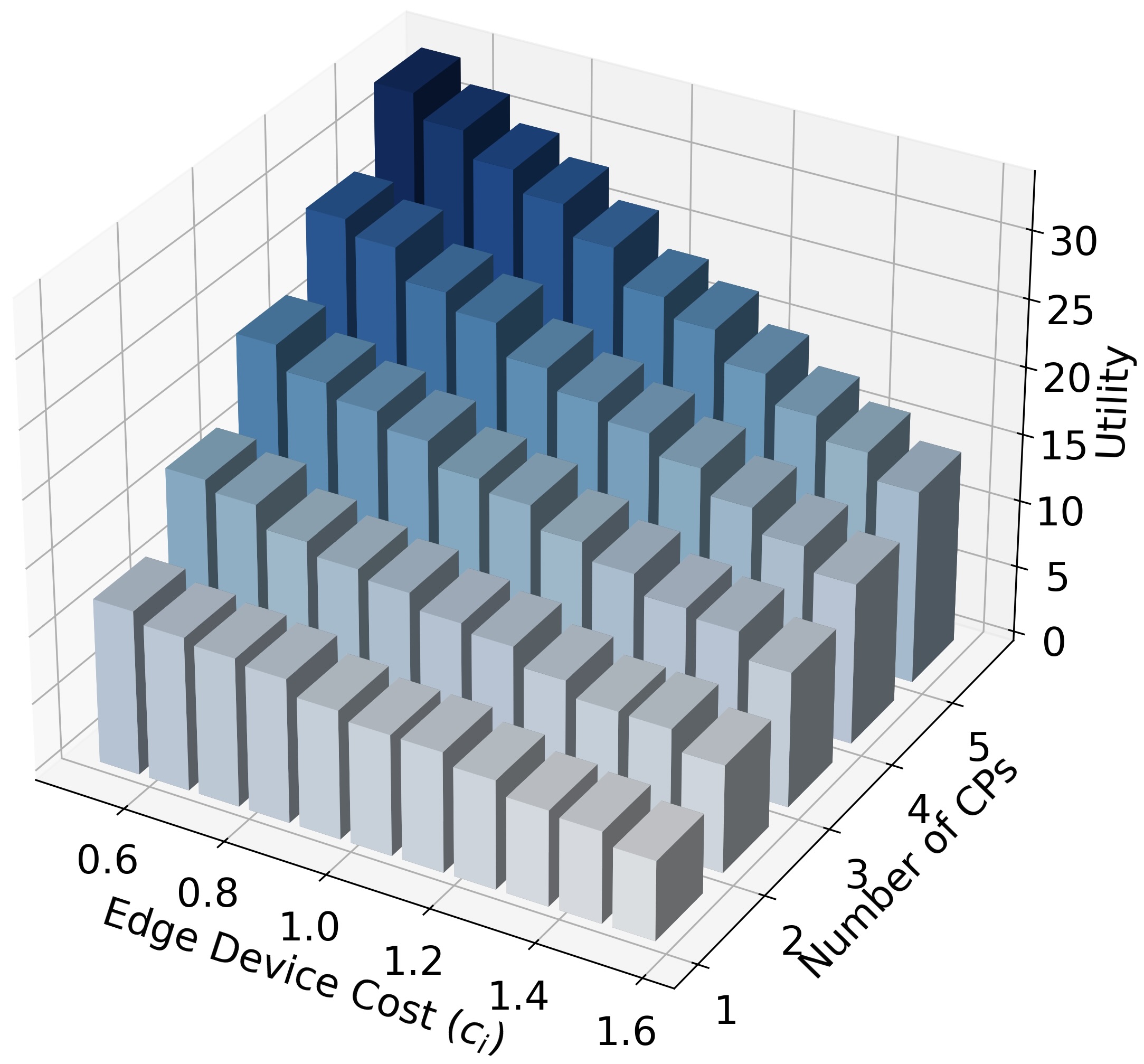}
    \caption{ Light storage constraint.}
    \label{plot:ed_light_c_i}
  \end{subfigure}
  \begin{subfigure}[b]{0.22\textwidth}
    \includegraphics[width=\linewidth]{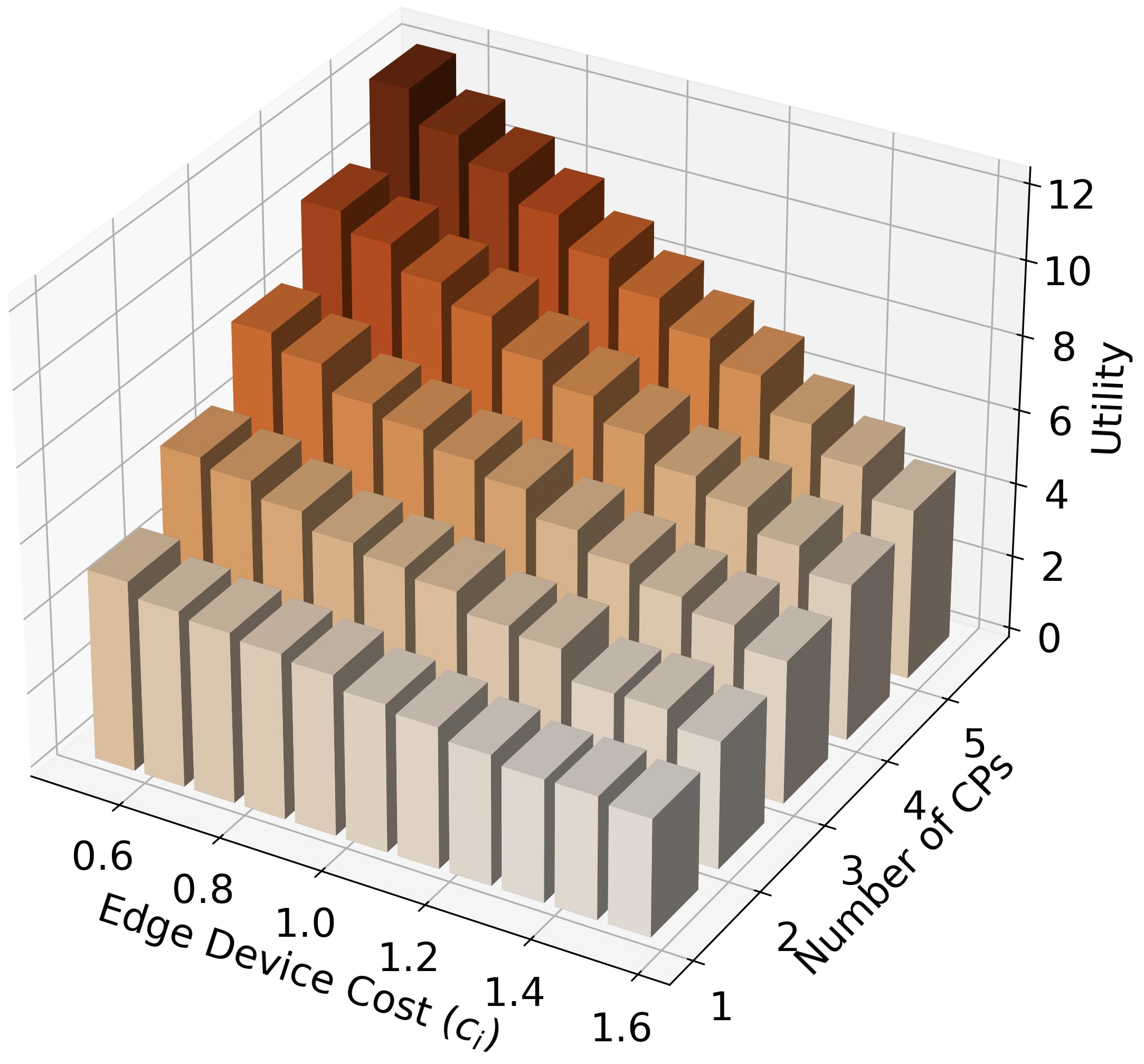}
    \caption{Strict storage constraint.}
    \label{plot:ed_strict_c_i}
  \end{subfigure}\hfill

  \caption{Average \gls{ed} utility vs. \gls{ed} cost parameter $c_i$ for different numbers of \glspl{cp}. 
  }
  \label{plot:c_i_impact_on_ed}
\end{figure}

\vspace{-1mm}
\section{Conclusions}\label{sec:conclusions}

This paper proposed a decentralized game-theoretic framework for multi-content-provider edge caching in mobile social networks, explicitly accounting for CP budget constraints, ED storage limitations, and strategic competition. By modeling CP–ED interactions as a Stackelberg game and CP competition as a non-cooperative game, we characterized equilibrium behavior and decentralized convergence under realistic system assumptions.
Under light storage constraints, CP competition was shown to form an exact potential game, guaranteeing equilibrium existence and uniqueness, and convergence. When storage constraints become binding, analytical guarantees no longer hold; however, extensive simulations demonstrated stable convergence across a wide range of network scales. In both regimes, convergence complexity was primarily driven by the number of competing CPs rather than the number of EDs, even when solver-based optimization was used.



\printbibliography

@article{luo2025cost,
  title={Cost-Effective Edge Data Caching With Failure Tolerance and Popularity Awareness},
  author={Luo, Ruikun and Zhang, Zujia and He, Qiang and Xu, Mengxi and Chen, Feifei and Dai, Xiaohai and Wu, Song and Jin, Hai},
  journal={IEEE Trans. on Mobile Computing},
  year={2025},
  publisher={IEEE}
}

@manual{ericsson,
    key = {report},
    title = {https://www.ericsson.com/en/reports-and-papers/mobility-report \\ /dataforecasts/mobile-traffic-forecast}
}

@article{zhang2025survey,
  title={A survey on privacy-preserving caching at network edge: Classification, solutions, and challenges},
  author={Zhang, Xianzhi and Zhou, Yipeng and Wu, Di and Sheng, Quan Z and Riaz, Shazia and Hu, Miao and Xiao, Linchang},
  journal={ACM Computing Surveys},
  volume={57},
  number={5},
  pages={1--38},
  year={2025},
  publisher={ACM New York, NY}
}

@article{chen2024dynamic,
  title={Dynamic task offloading and resource allocation for NOMA-aided mobile edge computing: An energy efficient design},
  author={Chen, Ying and Xu, Jiajie and Wu, Yuan and Gao, Jie and Zhao, Lian},
  journal={IEEE Trans. on Services Computing},
  volume={17},
  number={4},
  pages={1492--1503},
  year={2024},
  publisher={IEEE}
}

@article{cheng2024stackelberg,
  title={A Stackelberg Game Based Framework for Edge Pricing and Resource Allocation in Mobile Edge Computing},
  author={Cheng, Siyao and Ren, Tian and Zhang, Hao and Huang, Jiayan and Liu, Jie},
  journal={IEEE Internet of Things Journal},
  year={2024},
  publisher={IEEE}
}

@article{Xu2020,
  title={Game theory and reinforcement learning based secure edge caching in mobile social networks},
  author={Xu, Qichao and Su, Zhou and Lu, Rongxing},
  journal={IEEE Trans. on Information Forensics and Security},
  volume={15},
  pages={3415--3429},
  year={2020},
  publisher={IEEE}
}

@article{yu2025auction,
  title={Auction Theory and Game Theory Based Pricing of Edge Computing Resources: A Survey},
  author={Yu, Jiguo and Liu, Shun and Zou, Yifei and Wang, Guijuan and Hu, Chunqiang},
  journal={IEEE Internet of Things Journal},
  year={2025},
  publisher={IEEE}
}

@article{feng2025federated,
  title={Federated Deep Reinforcement Learning for Multimodal Content Caching in Edge-Cloud Networks},
  author={Feng, Weijia and Zuo, Xinyu and Zhang, Ruojia and Zhu, Yichen and Wang, Chenyang and Guo, Jia and Sun, Chuan},
  journal={IEEE Trans. on Network Science and Engineering},
  year={2025},
  publisher={IEEE}
}

@article{zyrianoff2024cache,
  title={Cache-it: A distributed architecture for proactive edge caching in heterogeneous iot scenarios},
  author={Zyrianoff, Ivan and Gigli, Lorenzo and Montori, Federico and Sciullo, Luca and Kamienski, Carlos and Di Felice, Marco},
  journal={Ad Hoc Networks},
  volume={156},
  pages={103413},
  year={2024},
  publisher={Elsevier}
}

@article{zhang2024cache,
  title={How to cache important contents for multi-modal service in dynamic networks: a DRL-based caching scheme},
  author={Zhang, Zhe and St-Hilaire, Marc and Wei, Xin and Dong, Haiwei and El Saddik, Abdulmotaleb},
  journal={IEEE Trans. on Multimedia},
  year={2024},
  publisher={IEEE}
}

@article{yang2018content,
  title={Content popularity prediction towards location-aware mobile edge caching},
  author={Yang, Peng and Zhang, Ning and Zhang, Shan and Yu, Li and Zhang, Junshan and Shen, Xuemin},
  journal={IEEE Trans. on Multimedia},
  volume={21},
  number={4},
  pages={915--929},
  year={2018},
  publisher={IEEE}
}

@article{niknia2025attention,
  title={Attention-Enhanced Prioritized Proximal Policy Optimization for Adaptive Edge Caching},
  author={Niknia, Farnaz and Wang, Ping and Wang, Zixu and Agarwal, Aakash and Rezaei, Adib S},
  journal={IEEE Trans. on Vehicular Technology},
  year={2025},
  publisher={IEEE}
}

@article{khan2024content,
  title={Content caching in mobile edge computing: a survey},
  author={Khan, Yasar and Mustafa, Saad and Ahmad, Raja Wasim and Maqsood, Tahir and Rehman, Faisal and Ali, Javid and Rodrigues, Joel JPC},
  journal={Cluster Computing},
  volume={27},
  number={7},
  pages={8817--8864},
  year={2024},
  publisher={Springer}
}

@article{liao2025context,
  title={Context-aware Proactive Edge Caching for Vehicular Edge Computing Based on Asynchronous Federated Learning},
  author={Liao, Zhuofan and Liu, Pang and Zheng, Bin and Tang, XiaoYong},
  journal={IEEE Internet of Things Journal},
  year={2025},
  publisher={IEEE}
}

@article{chaudhary2025pencache,
  title={PeNCache: Popularity based cooperative caching in Named Data Networks},
  author={Chaudhary, Pankaj and Hubballi, Neminath},
  journal={Computer Networks},
  volume={257},
  pages={110995},
  year={2025},
  publisher={Elsevier}
}

@article{wang2025investment,
  title={Investment-driven budget allocation and dynamic pricing strategies in edge cache network},
  author={Wang, Quyuan and Chen, Pengyang and Liu, Jiadi and Wang, Ying and Guo, Zhiwei},
  journal={Pervasive and Mobile Computing},
  volume={109},
  pages={102040},
  year={2025},
  publisher={Elsevier}
}

@article{yuan2024efficient,
  title={Efficient online computing offloading for budget-constrained cloud-edge collaborative video streaming systems},
  author={Yuan, Shijing and Liu, Yuxin and Guo, Song and Li, Jie and Chen, Hongyang and Wu, Chentao and Yang, Yang},
  journal={IEEE Trans. on Cloud Computing},
  year={2024},
  publisher={IEEE}
}

@article{guo2025optimal,
  title={Optimal Multi-Bitrate Video Caching and Processing in Edge Computing: A Stackelberg Game Approach},
  author={Guo, Min and Zhang, Di and Xing, Weiwei and Shao, Xun and Liu, Zhi and Zhang, Yaoxue},
  journal={IEEE Internet of Things Journal},
  year={2025},
  publisher={IEEE}
}

@article{yan2021pricing,
  title={Pricing-driven service caching and task offloading in mobile edge computing},
  author={Yan, Jia and Bi, Suzhi and Duan, Lingjie and Zhang, Ying-Jun Angela},
  journal={IEEE Trans. on Wireless Communications},
  volume={20},
  number={7},
  pages={4495--4512},
  year={2021},
  publisher={IEEE}
}

@article{doostmohammadi2023dynamic,
  title={Dynamic clustering for low-delay delivery of video content cached in MEC servers},
  author={Doostmohammadi, Ali and Khayyambashi, Mohammad Reza and Movahedinia, Naser and Becvar, Zdenek},
  journal={IEEE Systems Journal},
  volume={17},
  number={4},
  pages={5842--5853},
  year={2023},
  publisher={IEEE}
}

@article{abolhassani2024optimal,
  title={Optimal push and pull-based edge caching for dynamic content},
  author={Abolhassani, Bahman and Tadrous, John and Eryilmaz, Atilla and Y{\"u}ksel, Serdar},
  journal={IEEE/ACM Trans. on Networking},
  volume={32},
  number={4},
  pages={2765--2777},
  year={2024},
  publisher={IEEE}
}

@article{ismail2025survey,
  title={A survey on resource scheduling approaches in multi-access edge computing environment: a deep reinforcement learning study},
  author={Ismail, Ahmed A and Khalifa, Nour Eldeen and El-Khoribi, Reda A},
  journal={Cluster Computing},
  volume={28},
  number={3},
  pages={184},
  year={2025},
  publisher={Springer}
}

@article{fan2024contract,
  title={Contract theory and stackelberg game based storage resource allocation in edge caching systems},
  author={Fan, Yuqi and Zhang, Zhenghui and Hu, Zipeng and Wu, Weili and Du, Dingzhu},
  journal={IEEE Internet of Things Journal},
  year={2024},
  publisher={IEEE}
}

@article{jiang2022game,
  title={A game-theoretic analysis of joint mobile edge caching and peer content sharing},
  author={Jiang, Changkun and Gao, Lin and Luo, Jingjing and Zhou, Pan and Li, Jianqiang},
  journal={IEEE Trans. on Network Science and Engineering},
  volume={10},
  number={3},
  pages={1445--1461},
  year={2022},
  publisher={IEEE}
}

@article{zhong2020deep,
  title={Deep reinforcement learning-based edge caching in wireless networks},
  author={Zhong, Chen and Gursoy, M Cenk and Velipasalar, Senem},
  journal={IEEE Trans. on Cognitive Communications and Networking},
  volume={6},
  number={1},
  pages={48--61},
  year={2020},
  publisher={IEEE}
}

@article{wei2024cooperative,
  title={Cooperative caching algorithm for mobile edge networks based on multi-agent meta reinforcement learning},
  author={Wei, Zhenchun and Zhao, Yang and Lyu, Zengwei and Yuan, Xiaohui and Zhang, Yu and Feng, Lin},
  journal={Computer Networks},
  volume={242},
  pages={110247},
  year={2024},
  publisher={Elsevier}
}

@article{zhang2017computing,
  title={Computing resource allocation in three-tier IoT fog networks: A joint optimization approach combining Stackelberg game and matching},
  author={Zhang, Huaqing and Xiao, Yong and Bu, Shengrong and Niyato, Dusit and Yu, F Richard and Han, Zhu},
  journal={IEEE Internet of Things Journal},
  volume={4},
  number={5},
  pages={1204--1215},
  year={2017},
  publisher={IEEE}
}

@article{liu2017incentive,
  title={Incentive mechanism for computation offloading using edge computing: A Stackelberg game approach},
  author={Liu, Yang and Xu, Changqiao and Zhan, Yufeng and Liu, Zhixin and Guan, Jianfeng and Zhang, Hongke},
  journal={Computer Networks},
  volume={129},
  pages={399--409},
  year={2017},
  publisher={Elsevier}
}

@article{xu2017secure,
  title={Secure content delivery with edge nodes to save caching resources for mobile users in green cities},
  author={Xu, Qichao and Su, Zhou and Zheng, Qinghua and Luo, Minnan and Dong, Bo},
  journal={IEEE Trans. on Industrial Informatics},
  volume={14},
  number={6},
  pages={2550--2559},
  year={2017},
  publisher={IEEE}
}

@article{ma2024survey,
  title={A survey of ddos attack and defense technologies in multi-access edge computing},
  author={Ma, Yong and Liu, Long and Liu, Zhiquan and Li, Fagen and Xie, Qilin and Chen, Kaiwei and Lv, Chenyang and He, Ying and Li, Fan},
  journal={IEEE Internet of Things Journal},
  year={2024},
  publisher={IEEE}
}

@article{he2024design,
  title={Design and implementation of social based edge node selection algorithm},
  author={He, Qinlu and Wang, Rui and Zhang, Fan and Zhang, Xiang and Bian, Genqing and Zhang, Weiqi and Li, Zhen},
  journal={Multimedia Tools and Applications},
  volume={83},
  number={34},
  pages={81127--81149},
  year={2024},
  publisher={Springer}
}

@article{singleProvider,
  title={Secure Budget-Aware Edge Caching in Mobile Social Networks: A Dynamic Optimization Approach},
  author={Seyedi, Zahra and Sedghani, Hamta and Verticale, Giacomo and Passacantando, Mauro and Ardagna, Danilo},
  journal={Available at SSRN 5875412}
}

@article{sedghani2021,
title = {An incentive mechanism based on a Stackelberg game for mobile crowdsensing systems with budget constraint},
journal = {Ad Hoc Networks},
volume = {123},
pages = {102626},
year = {2021},
issn = {1570-8705},
author = {Sedghani, Hamta and Ardagna, Danilo  and Passacantando, Mauro  and Lighvan, Mina Zolfy  and Aghdasi, Hadi S. }
}

\clearpage
\appendix \label{sec:Appendix}

\subsection*{Summary of Notations}
\vspace{0.5cm}
    \begin{table}[H]
	\centering
	\caption{Summary of Notations}
	\label{tab:notations}
	\begin{tabular}{l l}
			\toprule
                \textbf{Notation} & \textbf{Description} \\
            \midrule
			\gls{cp} & Content provider \\
            \gls{ed} & Edge caching device \\
            \gls{mu} & Mobile user \\
            \gls{SCQ} & Secure Caching Quality  \\
            $\mathcal{I}$ & Set of $I$ \glspl{ed}\\
            $ed_i$ & \gls{ed} number $i$, for $i \in \{1, \dots, I\}$\\
            $\mathcal{O}$ & Set of $O$ \glspl{cp} \\
            $cp_o$ & \gls{cp} number $o$, for $o \in \{1, \dots, O\}$\\
			$\mathcal{M}_{o}$ & Set of contents of $cp_o$ \\
            $M_o$ & Number of contents of $cp_o$ \\
            $ct_{m}^o$ & content number $m$, for $m \in \{1, \dots, M_o \}$\\
			$N_i(t)$ & Number of MUs in coverage of $ed_i$ at \\ & time slot $t$\\
            $N_i(m, t)$ & Number of MUs requesting $ct_m$ in \\ & coverage of $ed_i$ at time slot $t$\\
			$f_m^o$ & Popularity of $ct_{m}^o$\\
			$p_m^o$ & Importance of $ct_{m}^o$\\
			$r_{i,m}^o$ & Ratio of MUs requesting $ct_{m}^o$ via $ed_i$\\
			$q_{i,m}^o$ & \gls{SCQ} service from $ed_i$ for $ct_{m}^o$ \\
			$g_{i,m}^o$ & \makecell[l]{Payment for secure caching of $ct_{m}^o$ on $ed_i$}  \\
            $ \delta_{i,m}^{o}$ & storage size of $ct_m$ \\
			$q_i$ & \gls{SCQ} strategy vector of $ed_i$ \\
			$g_i$ & Payment strategy vector for $ed_i$ \\
			$\mathbf{g}$ & Payment strategy matrix of the \gls{cp} \\
			$\mathbf{q}$ & \gls{SCQ} strategy matrix of all \gls{ed}s \\
			$c_{i}$  &Cost parameter of $ed_i$ with the highest \gls{SCQ}\\
			$\upsilon$ & Adjustment parameter for the $ed_i$\\
			$\psi _{i}$ & Resource consumption of $ed_i$ when it cheats\\
			$\theta$ & Payment adjust parameter.\\
			$\alpha$& Satisfaction parameter of secure \\ & content caching.\\
\bottomrule
		\end{tabular}
\end{table}

\vspace{0.5cm}

\subsection*{Proof of~\autoref{theorem:potential}}
We need to show that for any CP \( o \), if \( \mathbf{g}_o \) changes to \( \mathbf{g}_o' \) while \( \mathbf{g}_{-o} \) remains fixed, then:
\[
U_o(\mathbf{g}_o', \mathbf{g}_{-o}) - U_o(\mathbf{g}_o, \mathbf{g}_{-o}) = P(\mathbf{g}_o', \mathbf{g}_{-o}) - P(\mathbf{g}_o, \mathbf{g}_{-o})
\]
Let:
\begin{itemize}
	\item \( S = (\mathbf{g}_o, \mathbf{g}_{-o}) \): Original strategy profile.
	\item \( S' = (\mathbf{g}_o', \mathbf{g}_{-o}) \): New strategy profile after CP \( o \) changes.
\end{itemize}

\subsubsection*{Step 1: Compute Utility Difference} 

\

\[
U_o(S) = \sum_{i=1}^I \sum_{m=1}^{M_o} \left[A_{i,m}^{o} \log\left(1 + \frac{g_{i,m}^{o} \theta}{2 c_i \nu}\right) - \frac{(g_{i,m}^{o})^2 \theta^2}{2 c_i \nu} \right]
\]
\[
U_o(S') = \sum_{i=1}^I \sum_{m=1}^{M_o} \left[A_{i,m}^{o} \log\left(1 + \frac{g_{i,m}^{'o} \theta}{2 c_i \nu}\right) - \frac{(g_{i,m}^{'o})^2 \theta^2}{2 c_i \nu} \right]
\]
\begin{align}
	U_o(S') - U_o(S)& = \sum_{i=1}^I \sum_{m=1}^{M_o}  \left\{ 
    \left[ A_{i,m}^{o} \log\left(1 + \frac{g_{i,m}^{'o} \theta}{2 c_i \nu}\right) - \frac{(g_{i,m}^{'o})^2 \theta^2}{2 c_i \nu} \right] 
    \right.
    \nonumber \\
	& \quad - \left. \left[ A_{i,m}^{o} \log\left(1 + \frac{g_{i,m}^{o} \theta}{2 c_i \nu}\right) - \frac{(g_{i,m}^{o})^2 \theta^2}{2 c_i \nu} \right] 
    \right\}
\end{align}

\subsubsection*{Step 2: Compute Potential Difference}
\begin{align}
	P(S) = & \sum_{k \neq o} \sum_{i=1}^I \sum_{m=1}^{M_k} \left[A_{i,m}^{k} \log\left(1 + \frac{g_{i,m}^{k} \theta}{2 c_i \nu}\right) - \frac{(g_{i,m}^{k})^2 \theta^2}{2 c_i \nu} \right] \nonumber \\
	+ &  \sum_{i=1}^I \sum_{m=1}^{M_o} \left[ A_{i,m}^{o} \log\left(1 + \frac{g_{i,m}^{o} \theta}{2 c_i \nu}\right) - \frac{(g_{i,m}^{o})^2 \theta^2}{2 c_i \nu} \right]
\end{align}

\begin{align}
	P(S') = &  \sum_{k \neq o} \sum_{i=1}^I \sum_{m=1}^{M_k} \left[A_{i,m}^{k}  \log\left(1 + \frac{g_{i,m}^{k} \theta}{2 c_i \nu}\right) - \frac{(g_{i,m}^{k})^2 \theta^2}{2 c_i \nu} \right] \nonumber \\
	+& \sum_{i=1}^I \sum_{m=1}^{M_o} \left[A_{i,m}^{o}  \log\left(1 + \frac{g_{i,m}^{'o} \theta}{2 c_i \nu}\right) - \frac{(g_{i,m}^{'o})^2 \theta^2}{2 c_i \nu} \right] 
\end{align}

\begin{align*}
	P(S') - P(S) =& \sum_{k \neq o} \sum_{i=1}^I \sum_{m=1}^{M_k} \left[A_{i,m}^{k} \log\left(1 + \frac{g_{i,m}^{k} \theta}{2 c_i \nu}\right) - \frac{(g_{i,m}^{k})^2 \theta^2}{2 c_i \nu} \right] \nonumber \\
	+ &  \sum_{i=1}^I \sum_{m=1}^{M_o} \left[ A_{i,m}^{o} \log\left(1 + \frac{g_{i,m}^{'o} \theta}{2 c_i \nu}\right) - \frac{(g_{i,m}^{'o})^2 \theta^2}{2 c_i \nu} \right] \\
	- &  \sum_{k\neq o} \sum_{i=1}^I \sum_{m=1}^{M_k} \left[A_{i,m}^{k}  \log\left(1 + \frac{g_{i,m}^{k} \theta}{2 c_i \nu}\right) - \frac{(g_{i,m}^{k})^2 \theta^2}{2 c_i \nu} \right] \nonumber \\
	-& \sum_{i=1}^I \sum_{m=1}^{M_o} \left[A_{i,m}^{o}  \log\left(1 + \frac{g_{i,m}^{o} \theta}{2 c_i \nu}\right) - \frac{(g_{i,m}^{o})^2 \theta^2}{2 c_i \nu} \right] \\
	= & \sum_{i=1}^I \sum_{m=1}^{M_o}  \left[ A_{i,m}^{o} \log\left(1 + \frac{g_{i,m}^{'o} \theta}{2 c_i \nu}\right) - \frac{(g_{i,m}^{'o})^2 \theta^2}{2 c_i \nu} \right]  \\
	-& \sum_{i=1}^I \sum_{m=1}^{M_o}  \left[A_{i,m}^{o}  \log\left(1 + \frac{g_{i,m}^{o} \theta}{2 c_i \nu}\right) - \frac{(g_{i,m}^{o})^2 \theta^2}{2 c_i \nu} \right]\\
	= &  \sum_{i=1}^I \sum_{m=1}^{M_o} 
    \left\{ 
    \left[ A_{i,m}^{o} \log\left(1 + \frac{g_{i,m}^{'o} \theta}{2 c_i \nu}\right) - \frac{(g_{i,m}^{'o})^2 \theta^2}{2 c_i \nu} \right]
    \right.
    \\
	& \quad - \left. \left[ A_{i,m}^{o} \log\left(1 + \frac{g_{i,m}^{o} \theta}{2 c_i \nu}\right) - \frac{(g_{i,m}^{o})^2 \theta^2}{2 c_i \nu} \right] 
    \right\} \\
	=& U_o(S') - U_o(S)
\end{align*}

\end{document}